\documentclass[letterpaper]{article} %
\usepackage[]{aaai25}  %
\usepackage{times}  %
\usepackage{helvet}  %
\usepackage{courier}  %
\usepackage[hyphens]{url}  %
\usepackage{graphicx} %
\usepackage{natbib}  %
\usepackage{caption} %
\usepackage{algorithm}
\usepackage{algorithmic}

\usepackage{newfloat}
\usepackage{listings}
\DeclareCaptionStyle{ruled}{labelfont=normalfont,labelsep=colon,strut=off} %
\floatstyle{ruled}
\newfloat{listing}{tb}{lst}{}
\floatname{listing}{Listing}
\usepackage{stmaryrd}
\usepackage{amsthm}
\usepackage{amsmath}
\usepackage{cleveref}
\usepackage{mathtools}
\usepackage{enumitem}
\usepackage{xkeyval}
\usepackage{amssymb}
\usepackage{tikz}

\title{Contract-based Design and Verification of Multi-Agent Systems\\
with Quantitative Temporal Requirements}
\author {
    Rafael Dewes, %
    Rayna Dimitrova%
}
\affiliations {
    CISPA Helmholtz Center for Information Security\\ 
    \{rafael.dewes, dimitrova\}@cispa.de
}

{}
{}
{}

\DeclareMathOperator{\LTLu}{\mathcal{U}}

\providecommand{\LTLuntil}{\LTLu}

\providecommand{\LTLglobally}{\relax}
\providecommand{\LTLfinally}{\relax}
\providecommand{\LTLnext}{\relax}

\let\LTLglobally\LTLsquare
\let\LTLfinally\LTLdiamond
\let\LTLnext\LTLcircle

\newcommand{\power}[1]{\ensuremath{2^{#1}}}
\newcommand{\proj}[2]{\ensuremath{\mathsf{proj}(#1,#2)}}
\newcommand{\proje}[2]{\ensuremath{\mathsf{proj}_\exists(#1,#2)}}

\newcommand{\sema}[1]{\ensuremath{\llbracket}#1\ensuremath{\rrbracket}}

\newcommand{\nats}{\ensuremath{\mathbb N}}
\newcommand{\reals}{\ensuremath{\mathbb R}}

\newcommand{\mc}{\mathcal}
\newcommand{\gedc}{GEDC} %

\newcommand{\Distri}{\ensuremath{\mathit{Dec}}}

\newcommand{\distr}{\ensuremath{\mathbf{d}}}

\newcommand{\Hopeful}[1]{\ensuremath{\mathsf{Hopeful}(#1)}}
\newcommand{\hopef}{\ensuremath{H}}
\newcommand{\decset}{\ensuremath{D}}
\newcommand{\collab}[1]{\ensuremath{\mathsf{Collab}(#1)}}
\newcommand{\alldec}{\ensuremath{\mathit{ADec}}}
\newcommand{\agv}{\ensuremath{\mathit{AG}}}

\newcommand{\ltlf}{LTL\ensuremath{[\mathcal F]}}
\newcommand{\AP}{\ensuremath{\mathit{AP}}}
\newcommand{\true}{\ensuremath{\mathsf{true}}}
\newcommand{\false}{\ensuremath{\mathsf{false}}}

\newcommand{\Vals}[1]{\ensuremath{\mathit{Vals}(#1)}}

\newcommand{\inp}{\ensuremath{\mathit{I}}}
\newcommand{\outp}{\ensuremath{\mathit{O}}}

\newcommand{\init}{\ensuremath{\mathit{init}}}
\newcommand{\Outl}{\ensuremath{\mathit{Out}}}

\newcommand{\outnota}{\ensuremath{\outp_{\overline a}}}

\newcommand{\outa}{\ensuremath{\outp_{a}}}

\newcommand{\nota}{\ensuremath{\bar{a}}}

\newcommand{\Agents}{\ensuremath{\mathsf{Agt}}}

\newcommand{\comb}{\ensuremath{\mathit{comb}}}
\newcommand{\local}{\ensuremath{\mathit{local}}}
\newcommand{\shared}{\ensuremath{\mathit{shared}}}

\newtheorem{definition}{Definition}
\newtheorem{theorem}{Theorem}
\newtheorem{proposition}{Proposition}
\newtheorem{example}{Example}[section]

\newcommand{\atsite}{\ensuremath{\mathsf{at\_site}}}
\newcommand{\atbase}{\ensuremath{\mathsf{at\_base}}}
\newcommand{\call}{\ensuremath{\mathsf{call}}}
\newcommand{\badw}{\ensuremath{\mathsf{badw}}}

\begin{document}

\maketitle

\begin{abstract}
    Quantitative requirements play an important role in the context of multi-agent systems, where there is often a trade-off between the tasks of individual agents and the constraints that the agents must jointly adhere to. 
    We study multi-agent systems whose requirements are formally specified in the quantitative temporal logic LTL[$\mathcal{F}$] as a combination of local task specifications for the individual agents and a shared safety constraint. 
    The intricate dependencies between the individual agents entailed by their local and shared objectives make the design of multi-agent systems error-prone, and their verification time-consuming.
    In this paper we address this problem by proposing a novel notion of quantitative assume-guarantee contracts, that enables the compositional design and verification of multi-agent systems with quantitative temporal specifications.
    The crux of these contracts lies in their ability to capture the coordination between the individual agents to achieve an optimal value of the overall specification under any possible behavior of the external environment.
 We show that the proposed  framework improves the scalability and modularity of formal verification of multi-agent systems against quantitative temporal specifications.
\end{abstract}

\section{Introduction}\label{sec:intro}

Multi-agent systems (MAS) have become increasingly important in various domains, such as robotics, distributed computing, and autonomous vehicles. 
In MAS~\cite{Wooldridge09}, agents operate autonomously but may need to interact with each other to achieve shared objectives. 
Compared to larger monolithic systems, the ability of multiple agents to work collaboratively towards a common goal offers significant advantages in terms of efficiency and flexibility.

Coordination in MAS is a major challenge, as it involves not only the synchronization of actions among agents but also handles interdependencies and emergent conflicts. 
While it is crucial for aligning strategies, it can also introduce overhead and complexity, especially in dynamic and unpredictable environments. 
Additionally, achieving shared goals often requires balancing individual agent preferences and global objectives, which can introduce conflict. 

To address these challenges, \emph{contract-based design}~\cite{BenvenisteCNPRR18} offers a structured way to define and manage interactions between agents. 
Compositional design approaches employ different forms of contracts, such as assume-guarantee contracts~\cite{BenvenisteCFMPS07}. %
However, traditional contract frameworks  fall short in handling complex quantitative requirements and ensuring optimal collective performance, 
where agents must adapt to varying conditions and still aim for the best possible outcomes.

In this paper, we tackle the problem of designing an expressive contract framework for the compositional and best-effort satisfaction of quantitative specifications in MAS with collective goals. 
Our proposed framework extends traditional contract-based methods by incorporating quantitative measures and best-effort guarantees. 
Agents must strive for best possible outcomes even when perfect satisfaction %
is not feasible, 
ensuring that the system remains functional and effective in suboptimal environments. 
The approach includes formal verification methods to establish that the decomposition of global specifications into local contracts maintains the integrity of overall system goals, 
ensuring that individual agent behaviors align with the  collective objectives.

We consider compositional specifications in the quantitative temporal logic \ltlf~\cite{AlmagorBK16}, 
which extends linear temporal logic (LTL) with quality operators.
For a given trace, an \ltlf\ formula evaluates to a value from a finite set of possible satisfaction values.
We focus on specifications combining  requirements \emph{local} to individual agents and \emph{shared} requirements constraining the agents' interaction.
Local requirements only include actions of the respective agent, 
while shared objectives can contain any actions but must be safety properties.
Our framework ensures that individual agent behavior aligns with both local and shared system-wide objectives.
Existing work mostly studies the qualitative case where agents have a common objective.
In our model agents are cooperative, unlike non-cooperative MAS where concepts like Nash equilibria are central~\cite{KupfermanPV16}.
Closest to our approach is \cite{DewesD23}, however our framework is more expressive and,  unlike \cite{DewesD23}, our decomposition contracts are complete.
We will highlight these differences in \Cref{sec:contracts}.

Best-effort satisfaction means the highest satisfaction value enabled by the environment must be met by the system.
In contrast to previous work, our contracts take this into account to allow for different levels of coordination based on the inputs.
It relates to the qualitative notion of dominant strategies~\cite{DammF14} %
that perform as good as the best alternative.
We build on the idea of ``good-enough'' realizability~\cite{AlmagorK20} 
to encompass quantitative compositional specifications.%

The following example illustrates how the best-effort interpretation of compositional quantitative specifications leads to conditional obligations in a MAS.

\begin{example} \label{ex:intro} %
    In this example, we consider a MAS with three agents modelling vehicles: two of them (agents $a_1$ and $a_2$) can transport goods to a construction site,  and the third (agent $a_3$) provides a visual feed.
    The desired behavior of the overall system is specified as an \ltlf\ formula\\
    $\phantom{aaaaa}
        \Phi = \frac{1}{6} \varphi^1_\local \oplus \frac{1}{6} \varphi^2_\local \oplus \frac{1}{6} \varphi^3_\local \oplus \frac{1}{2} \Psi_\shared,
    $\\
   which is  the weighted sum ($\oplus$) of the sub-formulas $\varphi^i_\local$ and $\psi_\shared$.
    Each agent $a_i$ has a local specification $\varphi^i_\local$ that prescribes their individual objectives.
    The shared requirement $\Psi_\shared$ encodes a task assigned collectively to all agents. 
    Here, the local requirements are of equal priority, while the shared requirement $\Psi_\shared$ is weighed higher.
    
    The agents operate in an environment that can issue calls to move to the site and send weather warnings, modeled with propositions \call\ and \badw\, respectively.
    $\Psi_\shared$ requires that after each call, $a_3$ and at least one of $a_1$ and $a_2$ should respond in one or two time steps.
    According to $\varphi^i_\local$, vehicles $a_1$ and $a_2$ cannot travel from the base to the site in an instant, 
    and it is undesirable to travel during bad weather.
    We also require that each of $a_1$ and $a_2$ spends two consecutive time steps at the base infinitely often.
    The vehicle modelled by $a_3$ travels faster, and is unaffected by the weather, but it must be at the base infinitely often (to archive data). 
    
    Best-effort satisfaction of $\Phi$ means that it must be satisfied as well as possible under the current environment, by the system as a whole.
    Whenever \call\ and \badw\ occur together, to satisfy $\Psi_\shared$ either $a_1$ and $a_2$ must travel in bad weather, thus lowering the overall value of $\Phi$.
    If the frequency of this is high enough, then either both $a_1$ and $a_2$ must travel in bad conditions, 
    or one of them will be unable to stay at the base for two steps, violating $\varphi^i_\local$.

    There is an obvious need for coordination between the agents to collectively good-enough satisfy $\Phi$.
    Here, the agents must distribute the tasks to satisfy $\Psi_\shared$ and agree on values with which the local specifications for $a_1$ and $a_2$ are satisfied.
    Existing contract notions are not suitable for expressing this balancing of quantitative obligations.
    We present a framework for contract-based design that makes it possible to capture such coordination in the form of a contract.
    This enables independent design of the individual agents,  leaving flexibility for the implementation of each agent outside of their shared obligations.\looseness=-1
\end{example}

\section{Preliminaries}\label{sec:prelim}

In this section we recall definitions and notation necessary for presenting our framework, 
such as formal languages, the specification logic \ltlf, and the formalization of best-effort, or ``good-enough'', satisfaction.

\subsubsection{Languages}\label{sec:langs}
Let $\Sigma$ be a finite alphabet. 
The set of finite (infinite) words over $\Sigma$ is denoted by $\Sigma^*$ (resp. \ $\Sigma^\omega$).
For $\sigma  = \sigma_0\sigma_1\sigma_2\ldots\in \Sigma^\omega$, 
we denote $\sigma[i] = \sigma_i$,  and with $\sigma[i,\infty) = \sigma_{i},\sigma_{i+1},\ldots$ the suffix of $\sigma$ starting at position $i$.
For $\sigma' \in \Sigma^*$ and $\sigma'' \in \Sigma^* \cup \Sigma^\omega$, we denote with $\sigma' \cdot \sigma''$ their concatenation.
A language $L \subseteq \Sigma^\omega$ is  a \emph{safety language} if and only if for every $\sigma \in \Sigma^\omega \setminus L$ there exists a finite prefix $\sigma'$ of $\sigma$ such that $\sigma'\cdot\sigma'' \not \in L$ for every 
$\sigma'' \in \Sigma^\omega$.

For a set $X$, we denote with $\power{X}$ the powerset of $X$.
For a word $\sigma$ over alphabet $\power{X}$ and a subset $Y \subseteq X$ we denote with $\proj{\sigma}{Y}$ the projection of $\sigma$ onto the alphabet $\power{Y}$. 
For  $n \in\nats$, we define $[n]:=\{1,\ldots,n\}$. 
Given disjoint sets $X_1,\ldots,X_m$ and 
for each $i \in \{1,\ldots,m\}$ a word $\sigma_i$ over the alphabet $\power{X_i}$,  
we define the parallel composition $\parallel_{i=1}^m \sigma_i$ of $\sigma_1,\ldots,\sigma_m$ such that $(\parallel_{i=1}^m \sigma_i)[j] = \bigcup_{i=1}^m \sigma_i[j]$ for all $j$.

\subsubsection{The Specification Language \ltlf}\label{sec:ltlf}
We now recall the temporal logic \ltlf\ introduced in~\cite{AlmagorBK16}.
Let $\AP$ be a set of Boolean atomic propositions,  and 
$\mc F \subseteq \{f : [0,1]^k \to [0,1] \mid k \in \nats\}$ a set of functions.
The  \ltlf\ formulas are generated by the grammar
$\begin{array}{c}
\varphi ::= p \mid \true \mid \false \mid f(\varphi_1,\ldots,\varphi_k) \mid 
\LTLnext \varphi \mid\varphi_1 \LTLuntil \varphi_2,
\end{array}$
where $p \in \AP$ is an atomic proposition,  and $f \in \mc F$.

We consider sets $\mc F$ that include functions representing the Boolean operators,  i.e.,  $\{f_\neg,f_\land,f_\lor\} \subseteq \mc F$, where $f_\neg(x) := 1-x$, $f_\land(x,y) := \min\{x,y\}$ and $f_\vee(x,y) := \max\{x,y\}$.  
For ease of notation, we use the operators $\neg, \land,\lor$ instead of the corresponding functions.
One useful function is
the weighted average $x \oplus_{\lambda} y := \lambda \cdot x + (1-\lambda) \cdot y$, where $\lambda \in \{0,1\}$. 
Abusing notation,  we often use $\oplus$ with more than two arguments.
$\LTLnext$ and $\LTLuntil$ are the LTL \emph{next} and \emph{until} operators respectively~\cite{BaierKatoen08}.
We define the usual derived temporal operators 
\emph{finally}  $\LTLfinally \varphi := \true \LTLuntil \varphi$ and 
\emph{globally} $\LTLglobally \varphi  := \neg (\LTLfinally \neg\varphi)$.
$\LTLglobally$ requires its argument to hold at every step from now on;  
$\LTLfinally$ requires that it holds at some point in the future;
$\LTLnext$ requires it to be true in the next step.\looseness=-1

The semantics of \ltlf\ is defined with respect to words in $(\power{AP})^\omega$,  and maps an \ltlf\ formula $\varphi$ and a word $\sigma \in (\power{AP})^\omega$,  to a value $\sema{\varphi,\sigma} \in [0,1]$. 
For $f \in \mc F$,  we define $\sema{f(\varphi_1,\ldots,\varphi_k),\sigma} := f(\sema{\varphi_1,\sigma},\ldots,\sema{\varphi_k,\sigma})$. 
The semantics of the \emph{until} operator $\LTLuntil$ is  $\sema{\varphi_1 \LTLuntil \varphi_2,\sigma}:= \max_{i\geq 0}\{\min\{\sema{\varphi_2,\sigma[i,\infty)}, \min_{0\leq j< i}\sema{\varphi_1,\sigma[j,\infty)}\}\}.$
We refer the reader to~\cite{AlmagorBK16} for the full  definition of the semantics.
We denote with $\Vals{\varphi} := \{\sema{\varphi,\sigma} \mid \sigma \in (\power{AP})^\omega\}$ the set of possible values of an \ltlf\ formula $\varphi$.
\cite{AlmagorBK16} showed that $|\Vals{\varphi}|\leq \power{|\varphi|}$ for every \ltlf\ formula $\varphi$, where $|\varphi|$ is the formula's description size.  
Thus, each formula's set of possible values is finite.

\begin{example} \label{ex:spec} %

    We formalize \Cref{ex:intro} using \ltlf .
    The location and movement of the vehicles are modelled by propositions $\atbase_i$ and $\atsite_i$. 
    For $\varphi_\local^i$ for $i \in \{1,2\}$:
    \begin{align*}
        \varphi_\local^i := &\LTLglobally (\lnot \atbase_i \lor \lnot \atsite_i) \land \\
        &\LTLglobally (\atbase_i \to \LTLnext \lnot \atsite_i) \land \\
        &\LTLglobally \LTLfinally (\atbase_i \land \LTLnext \atbase_i) \land \\
        &\LTLglobally (\badw \to \LTLnext ( (\atbase_i \lor \atsite_i) \oplus_{\frac{1}{2}} \top )). 
    \end{align*}
    This formula evaluates to one of three values:
    It is $0$ if any of the first three conjuncts are violated. 
    Otherwise, it is $1$ if agent $a_i$ never travels during bad weather, 
    and $\frac{1}{2}$ if it does (even just once).
    Further, $\varphi^3_\local := \LTLglobally \LTLfinally \lnot \atsite_3$.

    $
    \Psi_\shared := \LTLglobally \big(\call \to \bigvee_{i=1}^2(
    \mathit{respond}_{i,3} \vee  \LTLnext\mathit{respond}_{i,3})\big),
    $
    where $\mathit{respond}_{i,3}:= \lnot \atsite_i \land \LTLnext \atsite_i \land \LTLnext\atsite_3$.
    Thus, vehicles must be at the site within two units of time.
    Vehicles $a_1$ and $a_2$ must arrive and not already be there, to model transporting goods.

\end{example}

\section{Quantitative Temporal  Specifications of Multi-Agent Systems}\label{sec:ge-def}

We now present a formal model of MAS and the notion of \emph{compositional \ltlf\ specifications} of such systems introduced in~\cite{DewesD23}.

\subsubsection{Formal Model for Multi-Agent Systems}
We consider MAS where agents interact both with each other and the external environment via Boolean variables,  i.e.,  atomic propositions.  
We fix a set $\AP$ of atomic propositions such that $\AP = \inp \uplus \outp$,  for a finite set $\inp$ of \emph{input atomic propositions} and 
a finite set $\outp$ of \emph{output atomic propositions} that are disjoint, i.e.,  $\inp \cap \outp = \emptyset$.
The input propositions $\inp $ are updated by the external environment and the output propositions $\outp$ are updated by the agents in the system.
We call words in $(\power{\inp  \cup \outp})^\omega$ \emph{(execution) traces} and those in $(\power{\inp})^\omega$ \emph{input traces}.

An \emph{agent} in the system is modelled as a Moore machine 
$M = (\inp_{M},  \outp_{M}, S,s^\init,\rho,\Outl)$, where $\inp_{M}$ and  $\outp_{M}$ are $M$'s sets of input and output propositions,  $S$ is a finite set of states, with initial state $s^\init \in S$, transition function $\rho : S \times \power{\inp_{M}} \to S$, and output labeling function $\Outl : S \to \power{\outp_{M}}$.
For input trace $\sigma_{\inp_M} \in (\power{\inp_M})^\omega$,  
$M$ produces the output trace 
$M(\sigma_{\inp_M}) \in(\power{\outp_M})^\omega$ such that 
$M(\sigma_{\inp_M})[i] = \Outl(s_i)$, 
where the sequence of states $s_0s_1\ldots \in S^\omega$ is such that 
$s_0 = s^\init$, and $s_{i+1} = \rho(s_i,\sigma_{\inp_M}[i])$ for all $i$.

A \emph{multi-agent system} (MAS) is a tuple $\mc S = (\inp,  \outp,\mc M)$, where $\inp$ and  $\outp$ are the sets of input and output propositions respectively, 
$\mc M = \langle M_1,\ldots M_n\rangle$ is a tuple of Moore machines representing the agents, where
$M_a = (\inp_{a},  \outp_{a}, S_a,s_a^{init},\rho_a,\Outl_a)$ is such that 
(1) for every $a,a'\in [n]$ with $a\neq a'$ it holds that $\outp_a \cap \outp_{a'} =\emptyset$,  (2) $\biguplus_{a=1}^n \outp_a = \outp$,  and (3) $\inp_a = \inp \cup (\outp \setminus \outp_a)$.
Conditions (1) and (2) state that the sets of output propositions of the individual agents partition $\outp$.  Condition (3) stipulates that each agent observes all input propositions $\inp$ and all output propositions of the other agents.  
We denote with $\Agents := [n]$ the set of all agents, and  with $\outnota  = \bigcup_{a' \in \Agents \setminus\{a\}} \outp_{a'}$ the set of outputs of agents different from $a$.
Our model uses a composition of Moore machines, similar to Moore synchronous game structures in~\cite{AlurHK02}, however here each machine models the implementation of an agent.
Given an input trace $\sigma_\inp\in (\power{\inp})^\omega$,  a MAS generates an output trace $\sigma_\outp\in (\power{\outp})^\omega$,  
which we denote by $(\parallel_{c=1}^n M_a)(\sigma_\inp)$, 
such that  $\proj{\sigma_\outp}{\outp_a} = M_a(\sigma_\inp \parallel \proj{\sigma_\outp}{\outnota})$ for all $a$.
That is,  $(\parallel_{a=1}^n M_a)(\sigma_\inp)$ is the \emph{composition of all the output traces} of the agents in $\mc S$.
We define the set of traces $\mc L (M_a) \subseteq (\power{\inp \cup \outp})^\omega$  as $\mc L (M_a) : = \{\sigma \in \power{\inp \cup \outp})^\omega \mid M_a(\proj{\sigma}{\inp_a}) = \proj{\sigma}{\outp_a}\}$.

\subsubsection{Good-Enough Satisfaction of \ltlf\ Specifications}
Following~\cite{AlmagorK20}, we consider \emph{good-enough satisfaction} of an \ltlf\ specification $\Phi$ by a system $\mc S$. Intuitively,  we require that for every input trace, the system's output results in the best value for $\Phi$ possible for this input trace.
To formalize this,  we use the notion of \emph{$v$-hopeful input sequences}~\cite{AlmagorK20}.
Those are the input sequences for which there exists output sequences such that $\Phi$ has value $v$ in the resulting executions.
Formally,  
for $v \in [0,1]$, 
the set of \emph{$v$-hopeful input sequences for the formula $\Phi$} is  
$\Hopeful{v, \Phi} := \{ \sigma_\inp \in (\power{\inp})^\omega \mid \exists \sigma_\outp \in (\power{\outp})^\omega. ~\sema{\Phi, (\sigma_\inp \parallel \sigma_\outp)} = v \}$.

\begin{definition}[Good-Enough Satisfaction]\label{def:spec-sat}
Given an \ltlf\ formula $\Phi$ over $\AP = \inp \uplus \outp$,  
and a MAS $\mc S = (\inp,  \outp,\mc M)$, 
we say that \emph{$\mc S$ satisfies $\Phi$}, denoted $\mc M \models \Phi$, 
if and only if for every $v \in [0,1]$ and every 
$v$-hopeful input sequence $\sigma_\inp \in \Hopeful{v, \Phi}$, it holds that 
$\sema{\Phi,\sigma_\inp \parallel \mc M(\sigma_\inp)} \geq v$.
\end{definition}

\subsubsection{Compositional \ltlf\ Specifications}
In this paper we focus on a particular class of \ltlf\ specifications  for MAS, called \emph{compositional specifications}, introduced in \cite{DewesD23}, which studies the problem of automatic synthesis of MAS from such specifications.
Compositional  specifications  are given as a combination of local specifications for the individual agents and a shared requirement.
Formally,   a \emph{compositional specification} is of the form 
$ \Phi  = \comb(\varphi_\local^1,\ldots,\varphi_\local^n,\Psi_\shared)$ where
\begin{itemize}
\item[(1)] $\comb : [0,1]^{n+1} \to [0,1]$ is  non-decreasing in each subset of arguments:  if $\comb(v'_1,\ldots,v'_{n+1}) < \comb(v_1,\ldots,v_{n+1})$, then,  $v_i' < v_i$ for some $i \in [n+1]$.
\item[(2)] The \emph{shared specification} $\Psi_\shared$ is a \emph{safety} \ltlf\ specification. 
This means that 
for every $\sigma\in (\power{\AP})^{\omega}$ and 
$v \in \Vals{\Psi_\shared}$,  
if $\sema{\Psi_\shared,\sigma} < v$, 
then there exists a prefix $\sigma'$ of $\sigma$, 
such that $\sema{\Psi_\shared,\sigma' \cdot \sigma''} < v$ for every infinite continuation $\sigma''\in (\power{\AP})^{\omega}$ of $\sigma'$.
\item[(3)] For each agent $a$, the \emph{local specification} $\varphi_\local^a$  for agent $a$ refers only to atomic propositions in $\outp_a \cup \inp$.
\end{itemize}
Such specifications occur  in MAS where we have task specifications for individual agents describing their desired behavior,  as well as a shared safety constraint.
Going forward, we write $\Phi$ to mean a compositional specification $ \Phi  = \comb(\varphi_\local^1,\ldots,\varphi_\local^n,\Psi_\shared)$.

\section{Quantitative Assume-Guarantee Contracts}\label{sec:contracts}

In this section we introduce a new notion of contracts for MAS with compositional \ltlf\ specifications as defined in \Cref{sec:ge-def}.
These contracts, termed \emph{good-enough decomposition contracts} (\gedc) can express coordination between the agents while taking into account the quantitative nature of \ltlf\ requirements.
In contrast to traditional contracts used for compositional design and verification,  the type of contracts we propose allows the specification of the quantitative contribution of each agent's local requirements to the satisfaction value of the overall specification. 
This,  in turn,  allows for assumptions and guarantees that impose weaker restrictions on agents' behavior by stating \emph{what} the satisfaction values of the local specifications should be, and not \emph{how} these values should be achieved.
After introducing \gedc , we define what it means for a MAS to satisfy a \gedc, and show that this notion is sound and complete.

A \gedc\ for a compositional \ltlf\ specification $\Phi$  consists of two components.
The first,  called \emph{allowed decompositions map},  specifies for each possible value $v$ of $\Phi$ a set of \emph{value decompositions} describing how the satisfaction of each local specification and the shared specification contribute to achieving value $v$ for $\Phi$. 
The second component of a \gedc\ associates with each value $v$ an \emph{assume-guarantee contract} that specifies,  in terms of behaviors,  how the agents cooperate to achieve the desired satisfaction values.\looseness=-1

We begin with the definition of value decompositions.
Intuitively,  the decompositions of $v$ with respect to $\Phi = \comb(\varphi_\local^1,\ldots,\varphi_\local^n,\Psi_\shared)$ are tuples of values for $\varphi_1,\ldots,\varphi_n,\Psi_\shared$ such that their combination results in $v$.
Formally,  given  $v \in \Vals{\Phi}$,  a \emph{value decomposition} of $v$ is a vector 
$\distr \in (\bigtimes_{a=1}^n \Vals{\varphi_\local^a} )\times \Vals{\Psi_\shared}$ where 
$\comb(\distr) = v$.
If $\distr = (\langle u_a \rangle_{a=1}^n, w)$,  we denote the  components of $\distr$ by 
$\distr[a] : = u_a$ for $a\in[n]$,  and $\distr[s] := w$.
For every $v\in \Vals{\Phi}$, we denote with $\Distri(\Phi,v) \subseteq \reals^{n+1}$ the set consisting of all possible value decompositions of $v$.  

Next, we define the notion of \emph{allowed decompositions map}  for $\Phi$.  
Intuitively, the set $\alldec$ associates with each $v \in \Vals{\Phi}$ the value decompositions allowed by the contract for $v$-hopeful input sequences.  
To ensure generality,  $\alldec$  allow different sets of decompositions to be associated with different subsets $H$ of $\Hopeful{v, \Phi}$.

\begin{definition} \label{def:a-decomp} 
    An \emph{allowed decompositions map} for $\Phi$ is a set $\alldec \subseteq \Vals{\Phi} \times \power{((\power{\inp})^\omega)} \times \power{\bigcup_{v \in \Vals{\Phi}} \Distri(\Phi,v)}$ where:
    \begin{enumerate}[label=(\roman*)]
            \item\label{cond:adec-inp} \textit{(Input coverage)} 
For every $v \in \Vals{\Phi}$ and $\sigma_\inp \in \Hopeful{v, \Phi}$,  there exists $(v',\hopef,\decset) \in \alldec$ such that $v'\geq v$ and $\sigma_\inp \in \hopef$. 
This ensures that every $v$-hopeful input sequence is covered by some set $\hopef$ for value $v' \geq v$. 
        \item\label{cond:adec-sat} \textit{(Satisfiable decompositions)}
For every $(v,\hopef,\decset) \in \alldec$,  $D \subseteq \Distri(\Phi,v)$ and for every $\distr \in \decset$ there exists $\sigma \in (\power{\inp \cup \outp})^\omega$
such that  $\sema{\Psi_\shared, \sigma} = \distr[s]$ and $\sema{\varphi_\local^a, \sigma} = \distr[a]$ for all $a \in [n]$.
That is,  $\decset$ only contains satisfiable decompositions from $\Distri(\Phi,v)$.
    \end{enumerate}
\end{definition}

For every input sequence $\sigma_\inp \in (\power{\inp})^\omega$ and $v \in \Vals{\Phi}$ we define 
$\alldec(v,\sigma_\inp) := \{\distr \mid\exists (v,\hopef,\decset) \in \alldec.~\sigma_\inp \in \hopef  \text{ and } \distr \in \decset\}$
to be the set of decompositions of $v$ allowed by $\alldec$ for the input sequence $\sigma_\inp$. 
Since every input sequence $\sigma_\inp$ is hopeful for some value in $\Vals{\Phi}$,
condition \ref{cond:adec-inp} in \Cref{def:a-decomp} implies that 
$\alldec(v,\sigma_\inp)$ is non-empty for some $v$.
This ensures that a system which conforms with
$\alldec$ has appropriate obligations w.r.t. every input sequence. 
Furthermore,  condition \ref{cond:adec-sat} guarantees that these obligations are met by satisfiable value decompositons.

For a given specification $\Phi$,  there may be multiple possible allowed decompositions maps that differ in the obligations they impose on the individual agents. 
Next,  we show the possible value decompositions for the formula in \Cref{ex:intro} and one possible allowed decompositions map.

\begin{example}\label{ex:adec} %
    The worst-case satisfaction value for  $\Phi$ in \Cref{ex:intro} is $\frac{5}{6}$.  
    If  \badw\ and \call\  are constantly true, this is the best that can be  achieved:
    $\Psi_\shared$ must be satisfied,  meaning that $a_1$ or $a_2$ must travel in bad weather, or not be at the base often enough.
    The four potential value decompositions for $v = \frac{5}{6}$ are $(\frac{1}{2},\frac{1}{2},1,1)$,  $(0,1,1,1)$,  $(1,0,1,1)$, $(1,1,0,1)$.
    The last tuple is impossible to achieve in the described environment.
$(\frac{1}{2},\frac{1}{2},1,1)$ corresponds to case where both $a_1$ and $a_2$ travel in bad weather. 
The other two, to either $a_1$ or $a_2$ not being at the base for two steps infinitely often.
   All three are possible for any $\frac{5}{6}$-hopeful input sequence.
    A most permissive $\alldec$ will thus contain $(\frac{5}{6}, \Hopeful{\frac{5}{6}},$$\{(\frac{1}{2},\frac{1}{2},1,1),$$(0,1,1,1), (1,0,1,1)\})$.  An alternative is to take a subset,  enforcing specific priorities.%
\end{example}

While allowed decompositions maps impose requirements on the contribution by each agent to the overall satisfaction value, 
they do not specify \emph{how} the agents cooperate to achieve these values.
This is done by providing for each $v \in \Vals{\Phi}$ an \emph{assume guarantee contract}  which specifies  what assumptions each agent makes on the behavior of other agents and what guarantees it provides in return. 
We recall the definition  from \cite{DewesD23}.
\begin{definition}\label{def:ag-contract}
An \emph{assume-guarantee contract} for $n$ agents is a tuple $\langle (A_a, G_a)\rangle_{a=1}^n$, where for every agent $a \in \Agents$:
\begin{itemize}
    \item $A_a, G_a \subseteq (\power{\AP})^\omega$ are safety languages defining the assumption and guarantee, respectively,  of agent $a$.
    \item $\bigcap_{a' \in \{1,\ldots,n\} \setminus \{a\}} G_{a'} \subseteq A_a$.
\item For every finite trace $\sigma \in (\power{\AP})^*$:
\begin{enumerate}[label=(\roman*)]
    \item\label{cond:ag1} If there exists an infinite word  $\sigma' \in (\power{\AP})^\omega$ such that $\sigma\cdot\sigma' \in A_a$,  then, for every $o_a \in \power{\outa}$,  there exists  $\sigma'' \in (\power{\AP})^\omega$ with $\sigma\cdot\sigma'' \in A_a$ and $\proj{\sigma''[0]}{\outa} = o_a$.\looseness=-1
    \item\label{cond:ag2} If there exists an infinite word  $\sigma' \in (\power{\AP})^\omega$ such that $\sigma\cdot\sigma' \in G_a$, then, for every $o_{\overline a} \in \power{\outnota}$,  there exists  $\sigma'' \in (\power{\AP})^\omega$ with $\sigma\cdot\sigma'' \in G_a$ and $\proj{\sigma''[0]}{\outnota} = o_{\overline a}$.\looseness=-1
\end{enumerate}
\end{itemize}
\end{definition}
\noindent
The above conditions ensure two properties.
Agent $a$ cannot violate its own assumption $A_a$ by selecting a bad output $o_a$. 
The remaining agents cannot violate the guarantee which agent $a$ must provide by selecting a bad output $o_{\overline a}$.

We are now ready to give the definition of a good-enough decomposition contract,  which pairs up an allowed decompositions map with a function mapping each $v \in \Vals{\Phi}$ to an assume-guarantee contract.
As a result,  for each $v \in \Vals{\Phi}$, the contract prescribes allowed decompositions,  depending on the $v$-hopeful input sequences, and an associated assume-guarantee  contract $\langle (A_a^v, G_a^v) \rangle_{a=1}^n$. 

\begin{definition}[Good-Enough Decomposition Contract]\label{def:decomp-contract}
    Let $\Phi$ be a compositional \ltlf\ specification over
$\AP = \inp \uplus \outp$, where 
$\outp = \outp_1,\ldots,\outp_n$ is a partitioning of $O$ 
among $n$ agents.
Let 
$\alldec$ be a map of allowed decompositions for $\Phi$, and let 
$\agv$ be a function that maps each $v \in \Vals{\Phi}$ 
to an assume-guarantee contract 
$\langle (A_a^v, G_a^v) \rangle_{a=1}^n$.
    We say that $(\alldec,\agv)$ is a \emph{good-enough decomposition contract} for $\Phi$ and $\outp_1,\ldots,\outp_n$ 
    if and only if 
    for every $v \in \Vals{\Phi}$ and 
    $\sigma_\inp \in  \Hopeful{v, \Phi}$,  
    there exist 
    $(v',\hopef,\decset) \in \alldec$,  $\distr \in \decset$ and
    $\sigma_\outp \in (\power{\outp})^\omega$, 
     such that $v' \geq v$ and the following hold:
    \begin{enumerate}[label=(\roman*)]
        \item \label{cond:dc-1} $\sigma_\inp \in \hopef$,  $\sema{\Psi_\shared, \sigma_\inp \parallel \sigma_\outp} = \distr[s]$, 
        \item \label{cond:dc-2} $\sema{\varphi_\local^a, \sigma_\inp \parallel \sigma_\outp} = \distr[a]$ for all $a \in [n]$, and
        \item \label{cond:dc-3} $(\sigma_\inp \parallel \sigma_\outp) \in \bigcap_{a \in \{1, \ldots, n\}} A_a^{v'}$.
    \end{enumerate}
\end{definition}

The conditions in \Cref{def:decomp-contract} ensure compatibility between $\alldec$ and $\agv$.
More concretely,  for each $v$-hopeful input sequence $\sigma_\inp$ there exists a corresponding output trace that agrees with some allowed value decomposition for value $v' \geq v$, 
and, in addition is consistent with the assumptions of all agents associated with value $v'$.
The next example illustrates the importance of these conditions. 

\begin{example}\label{ex:contract} %
    A \gedc\ for  $\Phi$ in \Cref{ex:intro}  combines some $\alldec$ with an assume-guarantee contract to capture the coordination between the agents.
    First, in input sequences where \call\ is true at least once, each agent must rely on help from other agents to cover $\Psi_\shared$.  More concretely,  agents $a_1$ and $a_2$ assume that $a_3$ will always respond.
    
    When the best possible value is $v=\frac{5}{6}$,
    $a_1$ and $a_2$ must coordinate whether one of them should respond to all calls,  
    or if both will reduce their local satisfaction value to $\frac{1}{2}$. 

    This is reflected in $\alldec$ and the assume-guarantee contract.  
    If  agent $a_1$ guarantees to cover all calls in bad weather for some input sequences,  and agent $a_2$  guarantees the same for all of the remaining input sequences,  then value decomposition $(\frac{1}{2},\frac{1}{2},1,1)$ is not attainable,  and both $(0,1,1,1)$ and $(1,0,1,1)$ must be in $\alldec$.
    Alternatively,  if the assumptions and guarantees force $a_1$ and $a_2$ to alternate responding to calls,  with both potentially travelling in bad weather, then $(\frac{1}{2},\frac{1}{2},1,1)$ is the only matching  value decomposition for  value $v=\frac{5}{6}$.

    Since a \gedc\ associates assumptions with each value $v$, it can enforce coordination for worst-case input sequences,  while giving the agents more freedom when the environment  is more helpful (allows a higher satisfaction value).
\end{example}

\emph{Remark.}
\gedc\ are more general than the contracts in \cite{DewesD23} in two key aspects.
First,  \gedc\ allow us to relate obligations to the hopefulness level of input sequences with respect to the overall specification,  instead of being independent of the global combined value.
Second, \gedc\ include explicit value decompositions which offer a more natural way to quantify the required contribution of each agent, as well as flexibility in specifying obligations.
This results in  a complete decomposition rule.

\smallskip

The value decompositions of an allowed decompositions map $\alldec$ in a \gedc\ specify obligations for each agent $a$ in terms of values for the local and the shared specification which $a$ must ensure.  
An agent $a$ should meet its obligation with respect to the value of $\Psi_\shared$ only if the other agents provide output sequences that permit that,  that is, if \emph{behaviors of the other agents are collaborative}. 
Otherwise,  $a$ is expected to achieve the best possible value for $\varphi_\local^a$. 
We call a sequence $\sigma_{\outnota} \in (\power{\outnota})^\omega$ 
\emph{$\decset$-collaborative for input sequence $\sigma_\inp$}, 
if there exists an output sequence 
$\sigma_{\outa} \in (\power{\outa})^\omega$ for agent $a$ and a value decomposition $\distr \in \decset$ where
\begin{itemize}
    \item $\sema{\varphi_\local^a, (\sigma_\inp \parallel \sigma_{\outnota} \parallel \sigma_{\outa})} \geq \distr[a]$, and
    \item $\sema{\Psi_\shared, (\sigma_\inp \parallel \sigma_{\outnota} \parallel \sigma_{\outa})} \geq \distr[s]$.
\end{itemize}
We denote by $\collab{\decset, \sigma_\inp}_{\nota}$ the set of all such sequences.

Intuitively,  $\decset$-collaborative behaviors of $\overline a$ are those that permit agent $a$ to provide output such that the resulting values for $\varphi_\local^a$ and $\Psi_\shared$ conform to at least one decomposition in $\decset$. 
This is illustrated by the following example.

\begin{example} %
In the setting of \Cref{ex:intro},  for input sequences where \call\ is always true and  \badw\ is always false,  a value $v=1$ for $\Phi$ is possible with $\decset = \{(1,1,1,1)\}$.
In this case,  the $\decset$-collaborative outputs for $a_1$ are those where agents $a_2$ and $a_3$ take care of \emph{enough} calls such that $a_1$ can achieve full satisfaction of $\varphi_\local^1$.
    More specifically, if $a_2$  sets $\atsite_2$ to true  every $k$ steps, that will give $a_1$ enough freedom to set $\atbase_1$ to true for two consecutive turns every $k$ steps,  and still meet its obligations towards $\Psi_\shared$.
\end{example}

We will now state what it means for a MAS $\mc S$ to satisfy a \gedc.
Intuitively,  the implementation $M_a$ of each agent $a$ should satisfy certain obligations for each $v \in \Vals{\Phi}$ 
and combination of $v$-hopeful input sequence $\sigma_\inp$ with outputs of the agents in $\overline a $ satisfying the assumption $A_a^v$.  
If the outputs of the agents in $\overline a $ are collaborative for $\sigma_\inp$ and $D$, the obligation is determined by the value distributions in $D$, and otherwise, by the best value of $\varphi^a_\local$ possible for $\sigma_\inp$.

\begin{definition}[\gedc\ Satisfaction]\label{def:contract-sat}
Let $\Phi$ be a compositional \ltlf\ specification over $\AP = \inp \uplus \outp$,  
$\outp_1,\ldots,\outp_n$ be a  partitioning of $O$, 
and let $(\alldec,\agv)$ be a \gedc\ for $\Phi$ and $\outp_1,\ldots,\outp_n$.
We say that a MAS $\mc S = (\inp,  \outp,\mc M)$ \emph{satisfies a \gedc\ $(\alldec,\agv)$} iff 
    for every agent $a \in \Agents$,  its implementation $M_a$ satisfies the following condition.
    
    For every 
     $\sigma_\inp \in (\power{\inp})^\omega $ and every 
    sequence of outputs of the other agents
    $\sigma_{\outnota} \in (\power{\outnota})^\omega$,
    for every $(v,\hopef,D) \in \alldec$ with 
    $\sigma_\inp \in \hopef$  and 
    $(\sigma_\inp \parallel M_a(\sigma_\inp \parallel \sigma_{\outnota})\parallel \sigma_{\outnota})\in A_a^v$, 
    where $\agv(v) = \langle (A_a^v, G_a^v)\rangle_{a=1}^n$,  
    \emph{all} of the following hold:
    
    \begin{enumerate}[label=(\roman*)]
    \item \label{eq:local-ge-1}
    If $\sigma_{\outnota} \in \collab{D,\sigma_\inp}_{\nota}$, then for all $\distr \in D$,
    
    \smallskip
    \noindent
    if there is an output sequence $\sigma_{\outp_a} \in (\power{\outp_a})^\omega$, with 
    \begin{itemize}
    \item $\sema{\varphi_\local^a, (\sigma_\inp \parallel \sigma_{\outp_a}\parallel \sigma_{\outnota})} \geq \distr[a]$, and
    \item $\sema{\Psi_\shared,(\sigma_\inp \parallel \sigma_{\outp_a} \parallel \sigma_{\outnota})} \geq \distr[s]$, 
    \end{itemize}
    
    \smallskip
    \noindent
    then,  there exists $\distr' \in D$ where\looseness=-1
    
    \begin{itemize}
    \item $\distr'[a'] = \distr[a']$ for all $a' \neq a$, and
    \item $\sema{\varphi_\local^a,(\sigma_\inp \parallel M_a(\sigma_\inp \parallel \sigma_{\outnota})\parallel \sigma_{\outnota})} \geq \distr'[a]$, 
    \item $\sema{\Psi_\shared,(\sigma_\inp \parallel M_a(\sigma_\inp \parallel \sigma_{\outnota}) \parallel \sigma_{\outnota})} \geq \distr'[s]$.
    \end{itemize}
    
    \item \label{eq:local-ge-2}
    If $\sigma_{\outnota} \in \collab{D,\sigma_I}_{\nota}$,  then the guarantee $G_a^v$ is satisfied,  that is, 
    $(\sigma_\inp \parallel M_a(\sigma_\inp \parallel \sigma_{\outnota}) \parallel \sigma_{\outnota})\in  G_a^v$.
    
    \smallskip
    
    \item \label{eq:local-ge-3}
    If $\sigma_{\outnota} \not\in \collab{D,\sigma_I}_{\nota}$, 
    then for $u \in \Vals{\varphi^a_\local}$,   
    if it holds that $\sigma_\inp \in \Hopeful{u,\varphi^a_\local}$,  then it also holds that
     $\sema{\varphi_\local^a,(\sigma_\inp \parallel M_a(\sigma_\inp \parallel \sigma_{\outnota})\parallel \sigma_{\outnota})} \geq u$.
    
    \end{enumerate}
    
\end{definition}

Generally,  for the contract to be satisfied,  the output $M_a(\sigma_\inp \parallel \sigma_{\outnota})$ produced by an agent $a$ must \ref{eq:local-ge-1} achieve best-effort satisfaction for $\varphi_\local^a$ and $\Psi_\shared$ depending on $\sigma_\inp$, 
in accordance with the decompositions $\distr$ contained in $\alldec$ and \ref{eq:local-ge-2} satisfy the guarantees $G_a^v$ as specified by the contract $(\alldec,\agv)$.
What section of the contract applies primarily depends on which set $\hopef$ the sequence $\sigma_\inp$ belongs to.
There are two exceptions to the above requirement,  depending on the  behavior $\sigma_{\outnota}$ of other agents $\nota$.
First, if some of the assumptions $A_a^v$ is violated,  agent $a$ is free from any obligations. 
Second,  as stated by \ref{eq:local-ge-3},  if $\sigma_{\outnota}$ is not collaborative,  agent $a$ must only maximize the satisfaction of $\varphi_\local^a$.

The next theorem states that {\gedc}s are \emph{sound}.  
That is,  a system that satisfies a \gedc\ $(\alldec,\agv)$ for a compositional \ltlf\ specification $\Phi$,  is guaranteed to satisfy $\Phi$.

\begin{theorem}[Soundness of \gedc]\label{thm:decomp-soundness}
Let $\Phi$ be a compositional \ltlf\ specification and $(\alldec,\agv)$ be a \gedc\ for $\Phi$.
Then,  for every MAS $\mc S = (\inp,  \outp,\langle M_1,\ldots,M_n\rangle)$ it holds that if $\mc S \models (\alldec,\agv)$,  then $\mc S \models \Phi$.
\end{theorem}
\emph{Proof Sketch.}
To show that $\mc S \models \Phi$, we show for every 
$\sigma_\inp \in \Hopeful{v, \Phi}$ that $\sema{\Phi,\sigma_\inp \parallel \mc M(\sigma_\inp)} \geq v$. 
To do so, we prove that for some $a \in \Agents$, the outputs of the other agents are $\decset$-collaborative for $\sigma_\inp$ and some $\decset$, and satisfy the assumption $A_a^v$.  Applying \ref{eq:local-ge-1} in \Cref{def:contract-sat} this yields $\sema{\Phi,\sigma_\inp \parallel \mc M(\sigma_\inp)} \geq v$.
To establish the existence of  agent $a$ with this property,  we use the 
fact that $\Psi_\shared$ is a safety \ltlf\ formula and that \ref{eq:local-ge-3} enforces maximizing the satisfaction value of the local specifications for non-collaborative inputs,  to show that the converse is impossible.
\qed 

{\gedc}s are also complete,  i.e.,  if $\mc S \models \Phi$,  then there exists a \gedc\ $(\alldec,\agv)$ for $\Phi$ with
$\mc S \models (\alldec,\agv)$.

\begin{theorem}\label{thm:decomp-completeness}
Let $\Phi$ be a compositional specification.
If $\mc S = (\inp,  \outp,\langle M_1,\ldots,M_n\rangle)$ is a MAS with $\mc S \models \Phi$, 
then there exists a \gedc\ $(\alldec,\agv)$  for $\Phi$ where  $\mc S \models (\alldec,\agv)$.\looseness=-1
\end{theorem}
\emph{Proof Sketch.}
$\mc S$ produces unique outputs for each input sequence,  the resulting trace satisfying exactly one value decomposition $\distr$.
$\alldec$ is chosen to contain exactly the finitely many value decompositions realized by $\mc S$.  The input traces matching each $\distr$ form the respective set $H$.
For $\agv$, we define each $G_a^v$ as the safety language encoding the implementation $M_a$, and the  assumption $A_a^v$ as  $\bigwedge_{a' \neq a} {G^v_{a'}}$.
Thus, $\agv$ prescribes the precise implementation $M_a$ of each agent $a$.
The full proof is in the supplementary material.
\qed 

\section{Compositional Verification}\label{sec:algo}

In this section we present an automata-theoretic method for checking contract satisfaction as formalized in \Cref{def:contract-sat}:
Given a compositional \ltlf\ specification $\Phi$,  
a \gedc\ $(\alldec,\agv)$ for $\Phi$, 
and a system $\mc S = (\inp,  \outp,\langle M_1,\ldots,M_n\rangle)$,
decide whether $\mc S \models (\alldec,\agv)$.
Further,  we show how to check whether a given pair $(\alldec,\agv)$ is a \gedc\ for $\Phi$,  that is, it satisfies the conditions of \Cref{def:a-decomp},  \Cref{def:ag-contract} and \Cref{def:decomp-contract}.

By \Cref{thm:decomp-soundness},  if $\mc S \models (\alldec,\agv)$, we can conclude that 
$\mc S \models \Phi$.  
We remark that the specification $\Phi$ can also be checked directly,  using standard techniques, by constructing the product of $M_1, \ldots, M_n$.
This, however, is often undesirable.  
Working with the full specification and the product system hampers scalability.
Furthermore,  a \gedc\ allows for the independent re-implementation of the individual agents as long as the conditions in  \Cref{def:contract-sat} are satisfied.

\subsubsection{Procedure for Checking Contract Satisfaction}
We first recall necessary automata definitions and constructions.

A \emph{generalized nondeterministic B\"uchi automaton} (GNBA) over an alphabet 
$\Sigma$ is a tuple $\mc{A} = (\Sigma, Q,\delta,Q_0,\alpha)$, where
$Q$ is a finite set of states, $Q_0 \subseteq Q$ is a set of
initial states, $\delta : Q \times \Sigma  \to 2^Q$ is the transition function, 
and $\alpha \subseteq 2^Q$ a set of sets of accepting states.
A run $\rho$ of a GNBA is accepting iff for every $F \in \alpha$, 
$\rho$ visits $F$ infinitely often.
A word $\sigma \in \Sigma^\omega$ is accepted by $\mc A$ if there exists an accepting run of $\mc A$ on $\sigma$.
The language of $\mc A$ is $\mc L(\mc{A}) := \{ \sigma \in \Sigma^\omega \mid \sigma \text{ is accepted by }\mc A\}$.  

In \cite{AlmagorBK16} it is shown that for any \ltlf\ formula $\varphi$ over $\AP$, and 
$V \subseteq \Vals{\varphi}$, 
there exists a GNBA $\mc{A}_{\varphi,V}$ such that 
$\mc L(\mc{A}_{\varphi,V})$ consists of the words $\sigma \in (\power{\AP})^{\omega}$ with $\llbracket \varphi, \sigma \rrbracket \in V$, and
$\mc{A}_{\varphi,V}$ has at most $2^{(|\varphi|^2)}$ states and at most $|\varphi|$ sets of accepting states.
We employ this construction to reduce the problem of checking contract satisfaction to checking emptiness of GNBA.

For the remainder of this section we assume that $(\alldec,\agv)$ is finitely represented using $\omega$-automata. 
More concretely,  the set $H$  in each $(v,H,D)\in \alldec$ is represented by a GNBA $\mc H$ where $\mc L(\mc H) = \{\sigma_{\inp} \parallel \sigma_\outp \in (\power{\AP})^\omega \mid \sigma_{\inp} \in H \}$.  
For each  $v \in \Vals{\Phi}$,  the assume-guarantee contract  
$\agv(v) = \langle (A_a^v, G_a^v)\rangle_{a=1}^n$  is represented by tuple of pairs of GNBA with alphabet $\power{\AP}$,  one for the assumption language $A_a^v$,  and one for the complement $\power{\AP}\setminus G_a^v$ of the guarantee.  We denote these automata by 
$\mc A_a^v$  and
$\overline{\mc G_a^v}$.
We assume that all automata in the contract are given as nondeterministic B\"uchi automata, i.e.,  have one accepting set.

In the next proposition we outline the construction of a GNBA $\mc C_a$ for agent $a$, whose language consists of the traces that violate some of the conditions in \Cref{def:contract-sat}.
The full construction is given in the supplementary material.

\begin{proposition}\label{prop:gnba-correctness}
For each agent $a \in \Agents$,  we can construct a GNBA $\mc C_a$ such that  for any $\sigma = \sigma_\inp \parallel \sigma_\outp \in (\power{\AP})^\omega$  we have that
$\sigma \in  \mc L(\mc C_a)$ if and only if there exists  $(v, \hopef, D) \in \alldec$ such that 
$\sigma_{\inp} \in H$,  and $\sigma \in A_a^v$, and at least one of the conditions in \Cref{def:contract-sat} is violated for agent $a$.
\end{proposition} 

\emph{Proof Sketch.}
From the formulas $\varphi_\local^a$ and $\Psi_\shared$, 
we construct for
$\sim\;\in\{<,\geq,=\}$,    
$u \in \Vals{\varphi^a_\local}$ and 
$w \in \Vals{\Psi_\shared}$ 
the GNBA
$\mc B_a^{\sim u}$, and $\mc B_s^{\sim w}$ such that 
$\sigma \in \mc L (\mc B_a^{\sim u})$ iff $\llbracket \varphi^a_\local, \sigma \rrbracket \sim u$, and similarly for $\mc B_s^{\sim w}$. 

Next, we construct a GNBA $\mc C_a$ that characterizes the language of traces $\sigma= \sigma_{\inp} \parallel \sigma_{\outp_a}\parallel \sigma_{{\outnota}}$ such that there exists $(v,\hopef,D)\in \alldec$ 
with $\sigma_{\inp} \in \hopef$,  and $\sigma \in A_a^v$,  and some of the conditions \ref{eq:local-ge-1}, \ref{eq:local-ge-2} or \ref{eq:local-ge-3} in \Cref{def:contract-sat} is violated.  
$\mc C_a$  is constructed from  $\mc B_a^{\sim u}, \mc B_s^{\sim w}, \mc A_a^v,  \overline{\mc G_a^v}$ and $\mc H$ using standard automata operations on GNBA.
\qed

Using the automata $\mc C_a$ constructed in \Cref{prop:gnba-correctness},  we can automatically check whether $\mc S \models (\alldec,\agv)$ by checking
for each agent $a \in \Agents$, 
whether the intersection of the languages 
of $\mc C_a$ and the 
implementation $M_a$ is empty.      
If  $\mc L (\mc C_a) \cap \mc L(M_a) = \emptyset$, 
then $M_a$ satisfies conditions of \Cref{def:contract-sat}, otherwise we find a counterexample.  
If for all $a \in \Agents$ we establish 
$\mc L (\mc C_a) \cap \mc L(M_a) = \emptyset$,  we conclude that $\mc S \models (\alldec,\agv)$, and hence $\mc S \models \Phi$.
Otherwise,  we have that $\mc S \not\models (\alldec,\agv)$. 
This, however, does not  imply $\mc S \not\models \Phi$, since the \gedc\ might be insufficient to verify $\Phi$.

\emph{Remark.} 
When checking whether $\mc L (\mc C_a) \cap \mc L(M_a) = \emptyset$,  this can be done separately for each value $v \in Vals(\Phi)$,  and incrementally when constructing $\mc C_a$. 
We further check emptiness on automata intersections in an on-the-fly manner.

The number of states of the  GNBA $\mc C_a$ is at most  
$ 2^{2^{\mc O(|\varphi_\local^a |^2 + | \Psi_\shared |^2)}}  \cdot \sum_{(v,\mc H, D) \in \alldec} |\mc H| \cdot |\mc A_a^v| \cdot | \overline{\mc G_a}^v| .$
Checking contract satisfaction  is done by  checking for each agent the emptiness of the intersection of  $\mc C_a$ and the respective implementation $M_a$,  which can be done in time linear in the size of their product~\cite{BaierKatoen08}.

\begin{theorem}[Checking \gedc\  Satisfaction]\label{thm:complexity-check-sat}
    Checking if a MAS $\mc S = (\inp,  \outp,\langle M_1,\ldots,M_n\rangle)$ 
    satisfies a \gedc\ $(\alldec,\agv)$ for a compositional \ltlf\ specification $\Phi$ %
    can be done for each agent $a \in \Agents$ in time linear in\\
    $|M_a| \cdot 2^{2^{\mc O(|\varphi_\local^a |^2 + | \Psi_\shared |^2)}}  \cdot \sum_{(v,\mc H, D) \in \alldec} |\mc H| \cdot |\mc A_a^v| \cdot | \overline{\mc G_a}^v|$.
\end{theorem}

\subsubsection{Procedure for Checking $(\alldec,\agv)$ is a \gedc\ for $\Phi$}
To check that $(\alldec,\agv)$ is a \gedc\ for $\Phi$, we construct the automaton for $\Phi$, of size at most $2^{(|\Phi|^2)}$. 
The conditions are established by checking  language inclusion and emptiness for GNBA obtained from those for $\Phi$ and $(\alldec,\agv)$.
Both can be done in time exponential in the size of the resulting automata,  resulting in a procedure that 
for each entry $(v,\hopef,\decset) \in \alldec$ runs in time exponential in the size of the automata $\mc H$,  $\mc A_a^v$ and $\overline{\mc G_a}^v$ and double exponential in $|\Phi|^2$.

\subsubsection{Modular Verification}
The contract-based approach allows for modular verification and design.
In a MAS that satisfies a given \gedc, the local modification of an agent's implementation or specification can be analyzed without considering the  full specification or implementation.

Consider a compositional  specification $\Phi$,   
a \gedc\ $(\alldec,\agv)$ for $\Phi$, 
and a MAS $\mc S = (\inp,  \outp,\langle M_1,\ldots,M_n\rangle)$ such that $\mc S \models (\alldec,\agv)$.
Suppose that the local specification for agent $a$ is modified from ${\varphi_\local^a}$ to  $\widehat\varphi_\local^a$, resulting in a revised global specification $\widehat\Phi$.
Let $\widehat M_a$ be a new implementation of agent $a$,  and 
$\widehat{\mc S} = (\inp,  \outp,\langle M_1,\ldots\widehat M_a,\ldots,M_n\rangle)$ the resulting MAS.
In order to check whether $\widehat{\mc S} \models (\alldec,\agv)$,  it suffices to verify only the condition in \Cref{def:contract-sat} pertaining to agent $a$,  with respect to the new implementation $\widehat M_a$ of agent $a$ and the original local specification $\varphi_\local^a$.  
If this is the case,  by \Cref{thm:decomp-soundness},  we can conclude that the modified system $\widehat{\mc S}$ satisfies the original specification $\Phi$. 
Otherwise,  the new implementation $\widehat M_a$ of agent $a$ is incompatible with the given contract.\looseness=-1

Considering the new local specification $\widehat\varphi_\local^a$, if $\widehat M_a$ satisfies the condition in \Cref{def:contract-sat} with respect to $\widehat\varphi_\local^a$,  we can still conclude that $\widehat M_a$ satisfies agent $a$'s obligation towards the remaining agents, but not necessarily $\widehat{\mc S} \models \widehat \Phi$.
The reason is that if $\widehat\varphi_\local^a$ imposes weaker constraints on the coordination between agents than $\varphi_\local^a$,  the old \gedc\ $ (\alldec,\agv)$ might be overly constraining and the overall system sub-optimal. 
This is illustrated in the next example.

\begin{example}\label{ex:relaxes-local-spec}
If the local requirements for an agent are  tightened in the design process, 
the agent's new specification might become incompatible with the existing contract, and thus require revision of the contract and the implementations of the other agents.
If, on the other hand, the designer relaxes the local requirements,  the satisfaction of the existing contract becomes easier. 
However,  the contract might not be a \gedc\ for the new specification $\widehat\Phi$, 
and the existing implementation of the MAS might no longer be optimal with respect to $\widehat\Phi$.
This is because in contrast to traditional requirements, where weaker requirements are always satisfied by a stricter implementation,  here we impose an optimality criterion.
To see this,  suppose that in the setting of \Cref{ex:intro}, the local specification for agent $a_1$ no longer  reduces in value if the vehicle travels in bad weather.
That is,  $\widehat\varphi_\local^1$ is obtained from $\varphi_\local^1$ by dropping the last conjunct.  Regardless of the given \gedc, any implementation for agent $a_1$ that conforms to the original contract clearly still satisfies the conditions of \Cref{def:contract-sat} with respect to $\widehat\varphi_\local^1$.
However,  if the given contract requires agent $a_2$ to travel during bad weather, the current implementation does not good-enough satisfy the new specification $\widehat \Phi$ and the contract is not a \gedc\ for $\widehat \Phi$.
The reason is that ${\widehat\varphi_\local^1}$ allows vehicle $a_1$ to travel in bad weather and still achieve local satisfaction of $1$,  making implementations where vehicle $a_2$ travels in bad weather sub-optimal.
Intuitively,  relaxing the local specification for $a_1$ enables the satisfaction of $\varphi^2_\local$ with a higher value than stipulated by the \gedc\ for $\Phi$, resulting in a higher possible satisfaction value for the new specification $\widehat\Phi$.
As the old contract permits implementations that are not good-enough with respect to $\widehat\Phi$, it needs to be revised.\looseness=-1
\end{example}

\section{Experimental Evaluation and Conclusion}\label{sec:conclusion}

To evaluate our framework,  we implemented the proposed compositional verification procedure in a prototype in Python using the Spot automata library~\cite{duret.22.cav}. 
We ran experiments on 6 MASs each with a number of agents between 3 and 6.
The results,  obtained on a consumer laptop with 16 GB of memory,  are shown in \Cref{tab:experiments}. 
The benchmarks are provided in the supplementary material.

\begin{table}[t] 
    \centering
    \scalebox{0.67}{
     \begin{tabular}{|l | r | r | r | r | r | r|} 
     \hline
     Example & $| \Agents |$ & $|\Phi|$ & $|\mc B_{\Phi} |$ & \multicolumn{3}{c |}{Runtime} \\
      & & & & Comp. & Rev. & Mono. \\
     \hline
     delivery\_vehicles & $3$ & $102$ & $1974$ & $329$ & $24$ & TO \\
     tasks\_collab & $4$ & $124$ & $349$ & $50$ & $4$ & $34$ \\
     tasks\_scaled & $6$ & $136$ & $547$ & $156$ & $13$ & $289$ \\
     robots\_help\_a3 & $3$ & $46$ & $86$ & $1.5$ & $0.2$ & $1.3$ \\
     robots\_help\_a5 & $5$ & $78$ & $470$ & $15.7$ & $3.1$ & $9.1$ \\
     synchr\_response & $3$ & $62$ & $759$ & $439$ & $10$ & TO \\
     \hline
     \end{tabular}
    }
     \caption{
        Runtime for compositional, local revision and monolithic in seconds, with timeout of 30 minutes.
        }\label{tab:experiments}
\end{table}

\subsubsection{Evaluation}
We evaluated the compositional verification (column  ``Comp.'') performed via checking \gedc\ satisfaction (including the verification that it is a \gedc\ for $\Phi$), and compared that to direct monolithic verification of $\Phi$ (column  ``Mono.'') by constructing the product system.
For the more complex examples, only the compositional algorithm finishes within the time limit,
as the automata grow too large, and are limited by memory.
In the compositional case, we need to handle more but smaller automata.
This addresses the memory consumption but involves more operations.
The overhead causes the monolithic approach to be faster when the specifications are small enough.
Overall, the modular approach scales better.
We report on the verification of local specification and implementation revision (column  ``Rev.''), which, as expected,  is faster than full verification.

\subsubsection{Conclusion}
We proposed a framework for modular design and analysis of multi-agent systems,  based on good-enough decomposition contracts.  It offers greater expressivity compared to existing notions of contracts,  and we show that it enhances scalability and modularity of verification.
More generally, our approach opens up a new direction in contract-based design of MAS, by extending it to the setting of quantitative specifications and good-enough satisfaction.

\bibliography{main}

\newpage

\appendix 

\section{Proofs from \Cref{sec:contracts}}
\setcounter{theorem}{0}
\begin{theorem}[Soundness of \gedc]
Let $\Phi$ be a compositional \ltlf\ specification and $(\alldec,\agv)$ be a \gedc\ for $\Phi$.
Then,  for every MAS $\mc S = (\inp,  \outp,\langle M_1,\ldots,M_n\rangle)$ it holds that if $\mc S \models (\alldec,\agv)$,  then $\mc S \models \Phi$.
\end{theorem}
\begin{proof}
    Let $\mc S = (\inp,  \outp,\mc M)$ be a MAS that satisfies a \gedc\ with the conditions of \Cref{def:contract-sat}.
    Let  $\widetilde\sigma_\inp \in (\power{\inp})^\omega$  be an input trace
    and let $v \in \Vals{\Phi}$ be a value such that 
    there exists $\widetilde \sigma_{\outp} \in  (\power{\outp})^\omega$ with 
    $\sema{\Phi,\widetilde\sigma_\inp \parallel\widetilde \sigma_\outp} = v$.
    Let us denote 
    $\begin{array}{lll}
    \widetilde u_a & := &  \sema{\varphi_\local^a, \widetilde\sigma_\inp \parallel \widetilde\sigma_\outp}$ for all $a \in \{1,\ldots, n\}\\
    \widetilde  w  &:=& \sema{\Psi_\shared, \widetilde\sigma_\inp \parallel \widetilde\sigma_\outp}
    \end{array}$.
    
    Furthermore,  let $\widetilde  \sigma := \widetilde\sigma_\inp \parallel \widetilde\sigma_\outp$.

   Let us denote with $\sigma'_\outp := (\parallel_{a=1}^n M_a)(\widetilde\sigma_\inp)$ the output trace produced by $\mc S$ on $\widetilde \sigma_\inp$, 
    and let $\sigma' := \widetilde\sigma_\inp \parallel \sigma'_\outp$.  
    
    To show that $\mc S\models \Phi$,  we need to show 
    $\sema{\Phi,\sigma'} \geq v$. 
    
    	By  the conditions in \Cref{def:decomp-contract} and the fact that $\widetilde\sigma_\inp \in \Hopeful{v, \Phi}$,  we can assume that we have chosen $\widetilde \sigma_\outp$ such that there exist
    	$(v',\hopef,\decset) \in \alldec$ and $\widetilde\distr \in \decset$ where
\begin{itemize}
	 \item $v' \geq v, \widetilde \sigma_\inp \in H$,  
     $\sema{\Psi_\shared,\widetilde\sigma} \geq \widetilde \distr[s]$, 
	\item  $\sema{\varphi_\local^a,\widetilde\sigma} \geq \widetilde \distr[a]$ for all  $a\in [n]$,  
	\item  $\widetilde\sigma \in \bigcap_{a \in \{1,\ldots, n\}}A_a$,  and
	\item $\proj{\widetilde \sigma}{\outnota} \in \collab{\decset,\widetilde \sigma_\inp}_{\nota}$ for all $a \in[n]$. 
 \end{itemize}

   We first establish a helpful property that states that  every prefix of $\sigma'$ has an infinite continuation for which the value of $\Phi$ is at least $v'$. 
    This is formalized in the following.
    
    \smallskip
    
    {\bf Property 1:}
For every $i \in \nats$,  for the prefix $\sigma'[0,i)$ of $\sigma'$ there exist an infinite continuation 
    $\sigma''_i$ such that
    \begin{itemize}
    \item $\sema{\Phi,(\sigma'[0,i))\cdot \sigma''_i} \geq v'$, 
    \item  $(\sigma' [0,i))\cdot \sigma''_i \in \bigcap_{a \in \{1,\ldots, n\}}A_a$ and
    \item $\proj{(\sigma'[0,i)\cdot \sigma''_i) }{\outnota} \in \collab{\decset,\widetilde \sigma_\inp}_{\nota}$ for all $a \in[n]$.
    \end{itemize}
    
    \noindent
    {\bf Proof of Property 1.} The proof is by induction on $i$.

	\smallskip
	\noindent
    {\it Base case $i=0$.}
    $\sigma'[0,i)$ is the empty sequence.
   Thus, $\sigma''_0 := \widetilde\sigma$ will have the desired properties by the choice of $\widetilde\sigma$.
	
	\smallskip
    \noindent
	{\it Induction step.} 
    Suppose that the claim holds for some $i \in \nats$,  and let us show that it is satisfied by $i+1$ as well.
    Let us order the agents $1,\ldots,n$.
    We will construct $\sigma''_{i+1}$ in $n$ steps in which we address the agents in the above order.
    
    At each step $a \in \{1,\ldots n\}$ we will use the conditions of \Cref{def:contract-sat} applied to agent $a$, 
    together with the fact that  $(\alldec,\agv)$ satisfies the properties of \Cref{def:decomp-contract}. 
    For the corresponding agent $a$ we will construct an intermediate trace $\sigma''_{i,a}$ where:
    \begin{itemize}
    \item $\sema{\Phi,(\sigma'[0,i))\cdot \sigma''_{i,a}} \geq  v'$,
    \item $(\sigma'[0,i))\cdot \sigma''_{i,a} \in \bigcap_{a \in \{1,\ldots, n\}}A_a$,
        \item $\proj{(\sigma'[0,i))\cdot \sigma''_{i,a}}{\outp_{\overline k}} \in \collab{\decset,\widetilde \sigma_\inp}_{\outp_{\overline k}}$ for $k \in [n]$.
    \item $\proj{\sigma''_{i,a}}{\outp_k}[0,i) = M_k((\widetilde\sigma_\inp \parallel \sigma'_{\outp_{\overline{k}}})[0,i))$ for all agents $k \leq a$.
    \end{itemize} 
    Once we have constructed $\sigma''_{i,a}$ for $a=n$,  we will define $\sigma''_{i+1} := \sigma''_{i,n}[1,\infty)$, which will satisfy the desired property for index $i+1$ by construction.

    We now describe the construction of $\sigma''_{i,a}$. 
    For convenience,  let us define $\sigma''_{i,0}  := \widetilde\sigma[i,\infty)$.
    For $a = 1,\ldots, n$ we apply conditions \ref{eq:local-ge-1} and \ref{eq:local-ge-2} from \Cref{def:contract-sat}  by instantiating:
    \begin{itemize}
    \item  $\sigma_\inp$ with $\widetilde\sigma_\inp$,
    \item  $v$ with the value $v'$,
    \item  $\sigma_\outp$ with $\widetilde\sigma_\outp$,
    \item $\distr$ with the value decomposition $\widetilde \distr$,
    \item $\sigma_{\outnota}$ with $\proj{(\sigma'[0,i)\cdot\sigma''_{i,a-1})}{\outnota}$,
    \item $\sigma_{\outp_a}$ with $\proj{( \sigma'[0,i)\cdot \sigma''_{i,a-1})}{\outp_a}$.
    \end{itemize}
    Thus,  by condition~\ref{eq:local-ge-1} in \Cref{def:contract-sat} we have that 
    \begin{align*}\sema{\Phi,
    \widetilde\sigma_\inp & \parallel 
    M_a\big(\widetilde\sigma_\inp \parallel \proj{(\sigma'[0,i)\cdot \sigma''_{i,a-1})}{\outnota})\big)\\&
     \parallel \proj{(\sigma'[0,i)\cdot \sigma''_{i,a-1})}{\outnota}} \geq  v'.
    \end{align*}
    Now,  let 
    \begin{align*}\sigma''_{i,a} := 
    \big(\widetilde\sigma_\inp & \parallel 
    M_a\big(\widetilde\sigma_\inp \parallel \proj{(\sigma'[0,i)\cdot \sigma''_{i,a-1})}{\outnota})\big)\\&
     \parallel \proj{(\sigma'[0,i)\cdot \sigma''_{i,a-1})}{\outnota}\big)[i,\infty).
    \end{align*}
    Since $M_a$ is an implementation of an agent represented  as a Moore machine,  we have that 
    \[M_a\big(\widetilde\sigma_\inp \parallel \proj{(\sigma'[0,i)\cdot \sigma''_{i,a-1})}{\outnota})\big)[0,i+1)  = 
    \sigma'[0,i+1).
    \]
    Therefore,  the above implies that 
    $\sema{\Phi,((\widetilde\sigma_\inp \parallel \sigma'_\outp)[0,i))\cdot \sigma''_{i,a}} \geq  v'$.
    Furthermore,  we have that $\proj{\sigma''_{i,a}}{\outp_a}[0,i) = M_a((\widetilde\sigma_\inp \parallel \sigma'_{\outnota})[0,i))$.
    
    What remains is to show that the trace defined in this way satisfies the assumptions.
    Since $M_a$ satisfies condition \ref{eq:local-ge-2} in \Cref{def:contract-sat}  we have that 
    $\sigma'[0,i)\cdot \sigma''_{i,a} \in G_a$.
    For $k \neq a$ we have $\sigma'[0,i)\cdot \sigma''_{i,a-1} \in G_k$.  
    Thus,  by the definition of guarantees, we can ensure that by potentially modifying the trace $\proj{\sigma'[0,i)\cdot \sigma''_{i,a}}{\outnota}$ we have 
    that $\sigma'[0,i)\cdot \sigma''_{i,a} \in G_k$ (since the difference between $\proj{\sigma''_{i,a}}{\outp_a}$ and $\proj{\sigma''_{i,a-1}}{\outp_a}$ happens only from position $i$ and later, 
    we can ensure that this does not affect the other desired properties of $\sigma''_{i,a}$).
    
    Therefore,  we have constructed $\sigma''_{i,a}$ with the desired properties.
This concludes the proof of {\bf Property 1} by induction.
 
\bigskip

    If $\Phi$ were a safety \ltlf\ specification,  {\bf Property 1} would be sufficient to establish the claim of the theorem. 
    Since this is not the case, we need an additional property pertaining to the local specifications (which are not necessarily safety properties).
    We use the additional condition \ref{eq:local-ge-3} in \Cref{def:contract-sat} to establish this. 
    It defines the requirements on an implementation for $a$ if $\proj{\sigma'_{\outp}}{\outnota} \not\in \collab{\decset,\widetilde\sigma_\inp}_{\nota}$.

\smallskip

We distinguish two cases.

    \smallskip
    \noindent
    {\bf Case 1:} For some agent $a$, the outputs of the other agents are $\decset$-collaborative for $\widetilde \sigma_\inp$ and satisfy the assumption $A_a$, that is,  $\sigma'_{\outnota} \in \collab{\decset,\widetilde\sigma_\inp}_{\nota}$ and $\sigma' \in A_a$.  
    In such a case, we can apply condition \ref{eq:local-ge-1} in \Cref{def:contract-sat} to establish that $\sema{\Phi, \sigma'} = v' \geq v$,  which concludes the proof in this case.
   
    \smallskip
    \noindent
    {\bf Case 2:} For every $a$, the outputs of the other agents are either not $\decset$-collaborative for $\widetilde \sigma_\inp$ or violate the assumption $A_a$.  
    Since $A_a$ is a safety language, and by {\bf Property 1} every prefix of $\sigma'$ has a  continuation that satisfies all the assumptions, the second case is not possible. 
    Thus,  for every $a$ it holds that $\sigma'_{\outnota} \not\in \collab{\decset,\widetilde\sigma_\inp}_{\nota}$.  
    Therefore,  by condition \ref{eq:local-ge-3} in \Cref{def:contract-sat} we have that for each $a$,  
    the trace $\sigma'$ achieves the maximal possible value for $\varphi_\local^a$ that is possible for $\widetilde \sigma_\inp$. 
    Therefore,  for every $a$ we have $\sema{\varphi_\local^a, \sigma'} \geq \sema{\varphi_\local^a, \widetilde\sigma} = \widetilde u_a$.
 
 Assume for the sake of contradiction that $\sema{\Phi,\sigma'} < v$. This means that $\sema{\Psi_\shared,\sigma'} < w$. 
 Since $\Psi_\shared$ is a safety property, there exists an index $i$, such that for every continuation of $\sigma'[0,i)$ the value of $\Psi_\shared$ is less than $w$.  
 Further, since $\Vals{\Psi_\shared}$ is finite, we can choose $i$ such that 
 for every continuation of $\sigma'[0,i)$ the value of $\Psi_\shared$ is less than or equal to
 $\sema{\Psi_\shared,\sigma'}$.
 By {\bf Property 1}, there exists a continuation $\sigma_i''$ where the value of $\Phi$ is at least $v' \geq v$.  
 By the above, for this continuation, we also have that the value of $\Psi_\shared$ is less than $w$ and  less than or equal to
 $\sema{\Psi_\shared,\sigma'}$.
Thus,  we have 

$\begin{array}{lll}
\sema{\Psi_\shared,\sigma'[0,i)\cdot\sigma_i''} \leq \sema{\Psi_\shared,\sigma'} & < & w\\
\sema{\Phi,\sigma'} & < & v\\
\sema{\Phi,\sigma'[0,i)\cdot\sigma_i''} & \geq & v.
\end{array}
 $
 
 Thus,  by the monotonicity property of $\comb$, we have that 
 for some $a$ it holds that   $\sema{\varphi_\local^a, \sigma'[0,i).\sigma''} > \sema{\varphi_\local^a, \sigma'}$,  which leads to a contradiction, since $\mc S $ is implemented via Moore machines. 
 Therefore, $\mc S$ can produce a different output only in the step after that, as input sequences are different.
\end{proof}

\begin{theorem}
Let $\Phi$ be a compositional \ltlf\ specification over atomic propositions  $\AP = \inp \uplus \outp$.
If $\mc S = (\inp,  \outp,\langle M_1,\ldots,M_n\rangle)$ is a multi-agent system such that $\mc S \models \Phi$, 
then there exists a good-enough decomposition contract $(\alldec,\agv)$  for $\Phi$ such that  $\mc S \models (\alldec,\agv)$.
\end{theorem}
\begin{proof}
    Let $\mc S = (\inp,  \outp,\mc M)$ be a multi-agent system that satisfies \Cref{def:spec-sat} for $\Phi$.
    For any input trace $\widetilde\sigma_\inp \in (\power{\inp})^\omega$, $\mc S$ produces an output $(\parallel_{a=1}^n M_a)(\widetilde\sigma_\inp) = \widetilde\sigma_\outp \in (\power{\outp})^\omega$ with some value $\sema{\Phi, (\widetilde\sigma_\inp \parallel \widetilde\sigma_\outp)} = \widetilde{v}$  for $\Phi$.
    The trace $(\widetilde\sigma_\inp \parallel \widetilde\sigma_\outp)$ realizes a specific decomposition $\widetilde\distr$ with $\comb(\widetilde\distr) = \widetilde{v}$.
    That is,  $\widetilde\distr$ is such that  for every $a \in [n]$, 
$\widetilde\distr[a] = \sema{\varphi_\local^a, (\sigma_\inp \parallel \widetilde\sigma_\outp)}$, and $\widetilde\distr[s] = \sema{\Psi_\shared, (\sigma_\inp \parallel \widetilde\sigma_\outp)}$.

As the number of possible values of $\Phi$ and the number of possible decompositions is finite,  there is a finite set of decompositions that are used by $\mc S$ for some input sequence.  
    Therefore we can construct $\alldec$ by adding for each value decomposition $\distr$ realized by $\mc S$ for some input sequence the triple  $(\comb(\distr), \hopef_{\distr}, \{\distr\})$, where
   $\hopef_{\distr} = \{ \sigma_\inp \in (\power{\inp})^\omega \mid \sema{( \langle\varphi_\local^a\rangle_{a=1}^n, \Psi_\shared),((\parallel_{a=1}^n M_a)(\sigma_\inp) \parallel \sigma_\inp)} \cong \distr\}$.

    For $\agv$ we define each $G_a^v$ as the safety language which encodes the implementation $M_a$ of an agent $a$,
    while the assumption $A_a^v$ is the conjunction $\bigwedge_{a' \neq a} {G}^v_{a'}$.
    These are the same for all $v$, as they must be sufficient for all input sequences and exactly the allowed decompositions in $\alldec$.
    $\agv$ therefore precisely captures the composed strategy $(\parallel_{a=1}^n M_a)$.

    Now, any input trace $\widetilde\sigma_\inp$ will be in exactly one language $\hopef_{\widetilde\distr}$ for exactly one decomposition $\widetilde\distr$.
    And this is precisely the one obtained by $(\widetilde\sigma_\inp \parallel \widetilde\sigma_\outp)$, where $\widetilde\sigma_\outp = (\parallel_{a=1}^n M_a)(\widetilde\sigma_\inp)$.
    We combine these to arrive at a good-enough decomposition contract $(\alldec,\agv)$, which satisfies its conditions~(\Cref{def:decomp-contract}), 
    and thus $\mathcal S$ satisfies the contract according to \Cref{def:contract-sat}.
\end{proof}

\newpage
\section{Details of Verification Methods from \Cref{sec:algo}}
\subsection*{Automata Constructions}

\subsubsection{Automata Representation of Decomposition Contracts}
We consider assumptions and guarantees represented by automata.
For each  $v \in \Vals{\Phi}$ we suppose that the  assume-guarantee contract $\agv(v) = \langle (A_a^v, G_a^v)\rangle_{a=1}^n$ is represented as follows. 
The assumption languages $A_a^v$, and the complement languages ${(\power{\AP})}^\omega\setminus G_a^v$ of the guarantees are given respectively as the GNBA with alphabet $\power{\AP}$:
$\mc A_a^v$  is the assumption of agent $a$ for value $v$.
$\overline{\mc G_a^v}$ is the complement of the guarantee of agent $a$ for value $v$.

Furthermore,  we assume that for each $(v,H,D)\in \alldec$,  the language $H$ is represented by an GNBA $\mc H$ where $\mc L(\mc H) = \{\sigma_\inp \parallel \sigma_\outp \in (\power{\AP})^\omega \mid \sigma_\inp \in H \}$.  
We will write $(v,\mc H,D)$ instead of $(v,H,D)$ for the elements of $\alldec$.  

\subsubsection{Automata for the Specifications}
For each of the \ltlf\  formulas $\varphi_\local^a$, $\Psi_\shared$ and $\Phi$,  we can apply the construction from \cite{AlmagorBK16} to construct the GNBA
$\mc B_a  :=  \mc A_{\varphi_\local^a,\Vals{\varphi_\local^a}}$,
$\mc B_s  :=  \mc A_{\Psi_\shared,\Vals{\Psi_\shared}}$ and
$\mc B_\Phi  :=  \mc A_{\Phi,\Vals{\Phi}}$.
Given  $\sim\;\in\{<,\geq,=\}$,    
$u \in \Vals{\varphi^a_\local}$,  
$w \in \Vals{\Psi_\shared}$ and
$v \in \Vals{\Phi}$, 
we  obtain automata  
$\mc B_a^{\sim u}$,  $\mc B_s^{\sim w}$, $\mc B_\Phi^{\sim v}$ with the following languages:
$\begin{array}{lll}
\mc L(\mc B_a^{\sim u}) & = & \{ \sigma \in (\power{\AP})^\omega \mid \sema{\varphi_\local^a,\sigma}  \sim u\}, \\
\mc L(\mc B_s^{\sim w} ) & = & \{ \sigma \in (\power{\AP})^\omega \mid \sema{\Psi_\shared,\sigma} \sim w \},\\
\mc L(\mc B_\Phi^{\sim v} ) & = & \{ \sigma \in (\power{\AP})^\omega \mid \sema{\Phi,\sigma} \sim v \}.
\end{array}$

\begin{proposition}\label{prop:gnba-correctness-appx}
For each agent $a \in \Agents$,  we can construct a GNBA $\mc C_a$ such that  for any $\sigma = \sigma_\inp \parallel \sigma_\outp \in (\power{\AP})^\omega$  we have that
$\sigma \in  \mc L(\mc C_a)$ if and only if there exists  $(v, \hopef, D) \in \alldec$ such that 
$\sigma_{\inp} \in H$,  and $\sigma \in A_a^v$, and at least one of the conditions in \Cref{def:contract-sat} is violated for agent $a$.
\end{proposition}
\begin{proof}
We construct a GNBA $\mc C_a$ 
that characterizes the language of traces $\sigma_\inp \parallel \sigma_{\outp_a}\parallel \sigma_{\outnota}$ such that  
there exists $(v,\hopef,D)\in \alldec$ such that $\sigma_\inp \in \hopef$, $\sigma \in A_a^v$, 
and some of the conditions \ref{eq:local-ge-1}, \ref{eq:local-ge-1} or \ref{eq:local-ge-3} in  \Cref{def:contract-sat} is violated. 

First,  from the formulas $\varphi_\local^a$ and $\Psi_\shared$,  
we construct the necessary GNBA
$\mc B_a^{\sim u}$ and $\mc B_s^{\sim w}$.

Then, for each $(v,\mc H,D)\in \alldec$ and $\distr \in D$,  from the GNBA
$\mc B_a^{\sim u}$ and $\mc B_s^{\sim w}$,  and the contract automata  $\mc A_a^v,  \overline{\mc G_a^v}$ and $\mc H$,  we construct  the  following GNBA: 
\begin{itemize}  
\item $\mc B_{\distr}' = \mc H \times \proje{\mc B_a^{= \distr[a]} \times \mc B_s^{= \distr[s]}}{\outa} \times \mc A_a^v$,\\
    which captures ``good'' inputs to agent $a$ with respect to $\distr$, those that enable $\distr$ and satisfy the assumption $\mc A_a^v$.

\item $\mc B_{\distr}'' = \big(\bigcup_{(u,w) \in \mathsf{Bad}(D,\distr[\overline a])} \mc B_a^{= u} \times \mc B_s^{= w}\big) \cup \overline{\mc G_a}^v$, \\
where $\mathsf{Bad}(D,\distr[\overline a]) \subseteq  \Vals{\varphi_\local^a} \times \Vals{\Psi_\shared}$ is such that $(u,w) \in\mathsf{Bad}(D,\distr[\overline a])$ if and only if 
for all $\distr' \in D$,  $\distr'[\overline a] \neq \distr[\overline a]$ or $u < \distr'[a]$ or  $w < \distr'[s]$. \\
    This automaton characterizes the ``bad'' outputs of agent $a$ with respect to $\distr$,
    either by not aligning with $\distr$ or violating its guarantee $\mc G_a^v$.
    Here, $\mathsf{Bad}(\decset,\distr[\overline a])$ enumerates the pairs of values that, from the perspective of $a$, do not conform to $\distr$.

\item $\mc B_{D}'''  = \mc H \times \overline{\bigcup_{\distr'\in D} \proje{\mc B_a^{= \distr'[a]} \times \mc B_s^{= \distr'[s]}}{\outa}}$. \\
    We use this automaton to represent outputs of other agents $\nota$ that are not $\decset$-collaborative for agent $a$, so not in $\collab{\decset,\sigma_\inp}_{\nota}$,
    for inputs $\sigma_\inp \in \hopef$.

\item $\mc B'''' = \bigcup_{u \in \Vals{\varphi_a}} (\proje{\mc B_a^{= u}}{\outa} \times \mc B_a^{< u})$, \\
    characterizing the traces that violate the obligation of best-effort satisfying the local specification $\varphi_\local^a$.

\end{itemize}

Intuitively,  the automaton $\bigcup_{\distr \in D} (\mc B_{\distr}' \times \mc B_{\distr}'')$ characterizes the language of the traces that violate condition \ref{eq:local-ge-1} or \ref{eq:local-ge-2} for $(v,H,D)$, 
and the automaton  $\mc B_{D}''' \times \mc B''''$ the language of the traces that violate condition \ref{eq:local-ge-3}.

Finally,   we construct the GNBA 
\[\mc C_a:= \bigcup_{(v,H,D)\in\alldec}\Big(\big(\bigcup_{\distr \in D} (\mc B_{\distr}' \times \mc B_{\distr}'')\big) \cup (\mc B_{D}''' \times \mc B'''') \Big).\]

Note that for a specific agent $a$, the two relevant entries of value decompositions $\distr \in \decset$ in terms of its obligations are $\distr[a]$ and $\distr[s]$.
These correspond to the satisfaction value of $\varphi_\local^a$ and $\Psi_\shared$, and we can enumerate just the unique pairs $\distr[a]$, $\distr[s]$ for $\decset$.
The resulting automaton using those pairs only will be equivalent to $\mc C_a$ as described above, 
as redundant pairs would not introduce any further obligations, and thus not change its language.

For the GNBA $\mc C_a$ constructed above, it holds that for $\sigma \in (\power{\AP})^\omega$  we have 
$\sigma = (\sigma_\inp \parallel \sigma_{\outa} \parallel \sigma_{\outnota})\in \mc L(\mc C_a)$ if and only if there exists  $(v, \hopef, D) \in \alldec$ such that 
$\sigma_\inp \in H$ and 
$\sigma \in A_a^v$, 
and at least one of the following is satisfied.
\begin{enumerate}
\item \label{app:autcon1} $\sigma_{\outnota} \in \collab{v,\sigma_\inp}_{\nota}$ , and  for some $\distr \in D$ 
 there is an output sequence $\sigma_{\outp_a}' \in (\power{\outp_a})^\omega$, with 
\begin{itemize}
\item $\sema{\varphi_\local^a, (\sigma_\inp \parallel \sigma_{\outp_a}'\parallel \sigma_{\outnota})} = \distr[a]$, and
\item $\sema{\Psi_\shared,(\sigma_\inp \parallel \sigma_{\outp_a}' \parallel \sigma_{\outnota})} = \distr[s]$.
\end{itemize}
and for every $\distr' \in D$ it holds that 
\begin{itemize}
\item $\distr'[a'] \neq \distr[a']$ for some $a' \neq a$,  or
\item $\sema{\varphi_\local^a,(\sigma_\inp \parallel \sigma_{\outa} \parallel \sigma_{\outnota})} < \distr'[a]$, or
\item $\sema{\Psi_\shared,(\sigma_\inp \parallel \sigma_{\outa}  \parallel \sigma_{\outnota})} < \distr'[s]$.
\end{itemize}
\item \label{app:autcon2} $\sigma_{\outnota} \in \collab{v,\sigma_I}_{\nota}$, 
$\sigma \not \in G_a^v$.
\item \label{app:autcon3}
 $\sigma_{\outnota} \not\in \collab{v,\sigma_I}_{\nota}$ and for some $u \in \Vals{\varphi_a}$,
there exists $\sigma_{\outp_a}' \in (\power{\outp_a})^\omega$ such that
$\sema{\varphi_\local^a,(\sigma_\inp \parallel \sigma_{\outp_a}'\parallel \sigma_{\outnota})} = u,$
but
$\sema{\varphi_\local^a,(\sigma_\inp \parallel \sigma_{\outa} \parallel \sigma_{\outnota})} < u$.
\end{enumerate}
The conditions are a direct consequence of the construction of the automata $\mc B_{\distr}',\mc B_{\distr}'', \mc B_{D}'''$ and $\mc B''''$.
$\mc C_a$ is the union automaton over all $(v, \hopef, \decset) \in \alldec$,
where either (\ref{app:autcon1}) for a $\distr \in \decset$ the inputs are ``good'' and the outputs ``bad'' (\ref{app:autcon2}),
or (\ref{app:autcon3}) the outputs of other agents $\sigma_{\outnota}$ are not in $\collab{v,\sigma_\inp}_{\nota}$ 
but outputs of $a$ $\outp_a$ are not ``good-enough'' with respect to $\varphi_\local^a$.
\end{proof}

\subsection*{Checking Contract Satisfaction}

\subsubsection*{Complexity Analysis}
The initial GNBA constructed from \ltlf\ formulas have at most $2^{(l^2)}$ states and at most $l$ sets of accepting states, where $l$ is the length of the respective formula~\cite{AlmagorBK16}.
We assume all given automata representations for the contract to be NBA, that is, they have exactly one set of accepting states.

The \emph{size} of an NBA/GNBA $\mc A$,  denoted $|\mc A|$, is equal to the number of states and transitions in the automaton.

For each $(v,\mc H,D)\in \alldec$ we establish the following.
\begin{itemize}
\item The automaton 
$\bigcup_{\distr \in D} (\mc B_{\distr}' \times \mc B_{\distr}'')$ has size at most
$|\Vals{\varphi_\local^a}| \cdot |\Vals{\Psi_\shared}|\cdot |\mc H| \cdot 2^{\mc O(|\varphi_\local^a |^2 + | \Psi_\shared |^2)} \cdot |\mc A_a^v| \cdot | \overline{\mc G_a}^v|$.
 Since the number of possible values is at most exponential in the length of the \ltlf\ formula,  the size is upper-bounded by
$|\mc H| \cdot 2^{\mc O(|\varphi_\local^a |^2 + | \Psi_\shared |^2)} \cdot |\mc A_a^v| \cdot | \overline{\mc G_a}^v|$.
In terms of accepting sets, $\mc B_{\distr}'$ has at most $|\varphi_\local^a| + |\Psi_\shared| + 1$, 
while $\mc B_{\distr}''$ has $|\varphi_\local^a| + |\Psi_\shared|$.
The union automaton $\bigcup_{\distr \in D} (\mc B_{\distr}' \times \mc B_{\distr}'')$ thus has at most $\mc O(|\varphi_\local^a| + |\Psi_\shared|)$ sets of accepting states. 

\item The automaton $\mc B_{D}'''$ requires complementation of the projected automata $\proje{\mc B_a^{= \distr'[a]} \times \mc B_s^{= \distr'[s]}}{\outa}$.
 We convert it into an NBA before complementation, 
 leading to an increase in size by factor $|\varphi_\local^a| + |\Psi_\shared|$~\cite{BaierKatoen08}, so we size have up to 
 $(|\Vals{\varphi_\local^a}| + |\Vals{\Psi_\shared}|) \cdot 2^{\mc O(|\varphi_\local^a |^2 + | \Psi_\shared |^2)}$.
 The complementation then leads to another exponential blow-up, resulting in size at most 
 $|\mc H| \cdot 2^{\mc O((|\Vals{\varphi_\local^a}| \cdot |\Vals{\Psi_\shared}| \cdot 2^{\mc O(|\varphi_\local^a |^2 + | \Psi_\shared |^2)})\cdot \log b)}$, 
 where $b = (|\Vals{\varphi_\local^a}| \cdot |\Vals{\Psi_\shared}| \cdot 2^{\mc O(|\varphi_\local^a |^2 + | \Psi_\shared |^2)})$ subsuming the previous increase.
 Since the number of possible values is at most exponential in the length of the \ltlf\ formula,   we obtain the upper bound 
 $|\mc H| \cdot 2^{2^{\mc O(|\varphi_\local^a |^2 + | \Psi_\shared |^2)}}$.
 Because $\mc H$ is also an NBA, $\mc B_{D}'''$ will only have one set of accepting states.
\end{itemize}

For agent $a$, the automaton $\mc B''''$ has size upper-bounded by
$2^{\mc O(|\varphi_\local^a |^2 + | \Psi_\shared |^2)}$, 
and number of sets of accepting states by $|\varphi_\local^a| + |\Psi_\shared|$.

Thus, the size of the GNBA $\mc C_a$ is at most  
\[ 2^{2^{\mc O(|\varphi_\local^a |^2 + | \Psi_\shared |^2)}}  \cdot \sum_{(v,\mc H, D) \in \alldec} |\mc H| \cdot |\mc A_a^v| \cdot | \overline{\mc G_a}^v| .\]
The number of sets of accepting states of $\mc C_a$ has an upper bound of $\mc O(|\varphi_\local^a| + |\Psi_\shared|)$.
There exists an NBA equivalent to $\mc C_a$  of size 
$\mc O(|\mc C_a| \cdot |\alpha|)$, where $|\alpha|$ is the number of  sets of accepting states in $\mc C_a$~\cite{BaierKatoen08}
Since $|C_a|$ is already doubly exponential in $|\varphi_\local^a |^2 + | \Psi_\shared|^2$,  the size of the NBA is upper-bounded by 
$2^{2^{\mc O(|\varphi_\local^a |^2 + | \Psi_\shared |^2)}}  \cdot \sum_{(v,\mc H, D) \in \alldec} |\mc H| \cdot |\mc A_a^v| \cdot | \overline{\mc G_a}^v| .$

Checking contract satisfaction can be done by checking language-emptiness on the intersection of the GNBA $\mc C_a$ and the respective implementation $M_a$ for each agent. 
This can be done in time linear in the size of the product of $M_a$ and the non-deterministic B\"uchi automaton obtained from $\mc C_a$~\cite{BaierKatoen08,VardiW94}.
Thus overall we have 
\[\sum_{a \in \Agents} |M_a| \cdot 2^{2^{\mc O(|\varphi_\local^a |^2 + | \Psi_\shared |^2)}}  \cdot \sum_{(v,\mc H, D) \\ \in \alldec} |\mc H| \cdot |\mc A_a^v| \cdot | \overline{\mc G_a}^v|.\]

\setcounter{theorem}{2}
\begin{theorem}[Complexity of Checking Decomposition Contract Satisfaction]
    Checking if a multi-agent system $\mc S = (\inp,  \outp,\langle M_1,\ldots,M_n\rangle)$ 
    satisfies a decomposition contract $(\alldec,\agv)$ for a compositional \ltlf\ specification $\Phi$
    can be done in time linear in 
    \[\sum_{a \in \Agents} |M_a| \cdot 2^{2^{\mc O(|\varphi_\local^a |^2 + | \Psi_\shared |^2)}}  \cdot \sum_{(v,\mc H, D) \\ \in \alldec} |\mc H| \cdot |\mc A_a^v| \cdot | \overline{\mc G_a}^v|.\]
\end{theorem}

\subsection*{Checking Validity of Given Contract}

We now describe a procedure for checking whether a given pair $(\alldec,\agv)$ is a good-enough decomposition contract w.r.t. a compositional \ltlf\ specification $\Phi$.
This includes checking that $\alldec$ satisfies the conditions of  \Cref{def:a-decomp}, and that for each $v\in \Vals{\Phi}$,  $\agv(v)$ satisfies the conditions of \Cref{def:ag-contract}.  
This is done as follows: 
\begin{itemize}
\item Checking the conditions of \Cref{def:a-decomp} for $\alldec$:
\begin{itemize}
	\item[\ref{cond:adec-inp}] \textit{(Input coverage)}  
Checking that every $v$-hopeful input sequence is covered by some set $\hopef$ for value $v' \geq v$	can be done by verifying for every $v \in \Vals{\Phi}$ that
$\mc L(\proje{\mc B_\Phi^{= v}}{\outp}) \subseteq  \mc{L}(\bigcup_{(v',\mc H, \decset)  \in \alldec:v' \geq v} \mc H)$,
where the GNBA $\mc B_\Phi^{= v}$ is constructed from  $\Phi$.
	\item[\ref{cond:adec-sat}] \textit{(Satisfiable decompositions)} can be verified by checking that $\mc L(\bigtimes_{a=1}^n \mc B_a^{= \distr[a]} \times \mc B_s^{= \distr[s]}) \neq \emptyset$ for every value decomposition $\distr$ appearing in $\alldec$.
\end{itemize}
\item Checking the conditions of \Cref{def:ag-contract} for $\agv$:

The language inclusion condition can be checked using the respective automata,  or ensured by constructing the guarantees based on the assumptions.  
The other conditions are ensured by requiring that the assumptions and guarantees are represented by safety automata that syntactically satisfy \ref{cond:ag1}  and \ref{cond:ag2} from \Cref{def:ag-contract}.

\item Checking the conditions of \Cref{def:decomp-contract} of \gedc:

The conditions are satisfied if for every $v\in \Vals{\Phi}$, 
the language $\mc L(\proje{\mc B_{\Phi}^{=v}}{\outp} )$ 
is a subset of 
$\mc L(\bigcup_{(v',H,D) \in \alldec:v' \geq v} \bigcup_{\distr \in D} (\mc H \times \proje{\mc E_{v',\distr}}{\outp}))$,
 where the automaton $\mc E_{v',\distr}$ is defined by\\
$
\mc E_{v',\distr}  : =  (\bigtimes_{a=1}^n \mc B_a^{= \distr[a]}) \times \mc B_s^{= \distr[s]} \times (\bigtimes_{a=1}^n \mc A_a^{v'}). 
$ 
\end{itemize}

\subsubsection*{Complexity Analysis}

\begin{itemize}
\item Checking input coverage is performed by checking language inclusion for each $v \in \Vals{\Phi}$,  that is, by checking emptiness of the language of the product of $\proje{\mc B_\Phi^{= v}}{\outp}$ and 
$\bigtimes_{(v',\mc H, \decset)  \in \alldec:v' \geq v} \overline{\mc H}$.
Since $|\Vals{\Phi}|$ is at most exponential in $|\Phi|$,
the overall time is linear in $2^{\mc O(|\Phi|^2)} \cdot \Pi_{{(v',\mc H, \decset)  \in \alldec:v' \geq v}} 2^{|\mc H|}$.
\item Checking the satisfiability of the value decompositions in $\alldec$ is performed by checking language emptiness for each $\distr \in \alldec$.
The respective automaton has a size of at most 
$2^{\mc O(\sum_{a=1}^n |\varphi_\local^a |^2 + | \Psi_\shared |^2)}$, 
 or  $2^{\mc O(|\Phi|^2)}$.
Thus, the overall time  is linear in $n_{dec} \times 2^{\mc O(|\Phi|^2)}$, where $n_{dec}$ is the number of value decompositions in $\alldec$.
\item Checking the requirements on $\agv$ is performed by checking language inclusion for each $v \in \Vals{\Phi}$ and $a \in \Agents$. 
Each inclusion check can be done in time exponential in 
$|\mc A_{a}^v|\cdot \Pi_{a' \neq a} |\overline{\mc G_{a'}^v}|$~\cite{SistlaVW85}.
Since $\Vals{\Phi} \leq  2^{|\Phi|}$, the total time is bounded by 
$2^{|\Phi|} \cdot |\Agents| \cdot 2^{|\mc A_{a}^v|\cdot \Pi_{a' \neq a} |\overline{\mc G_{a'}^v}|}$.
\item Checking the conditions of \Cref{def:decomp-contract} is performed by checking language inclusion for each $v \in \Vals{\Phi}$.
The GNBA $\mc B_{\Phi}^{=v}$ has at most $2^{|\Phi|^2}$ states and $|\Phi|$ sets of accepting states.
We construct the complement of the automaton 
$\bigcup_{(v',H,D) \in \alldec:v' \geq v} \bigcup_{\distr \in D} (\mc H \times \proje{\mc E_{v',\distr}}{\outp})$ by complementing $\mc H$ and  $\proje{\mc E_{v',\distr}}{\outp}$,  
resulting in a complement automaton of size at most $\Pi_{(v',H,D) \in \alldec:v' \geq v} (2^{\mc O(|\mc H | \cdot \log |H|)} +   2^{\mc O((|D| \cdot 2^{|\Phi|^2} \cdot \Pi_{a \in \Agents} |\mc A_a^{v'}|) \cdot \log (|D| \cdot 2^{|\Phi|^2} \cdot \Pi_{a \in \Agents} |\mc A_a^{v'}|))})$.
\end{itemize}

\subsection*{Modular Verification}

\begin{example}\label{ex:stronger-local-spec}
    We consider a case where the requirements are tightened in the design process.
    Take the setting from \Cref{ex:intro}, with a contract of $a_1$ and $a_2$ forcing alternation every two steps in good weather when  \call\ is always true.
    The local specification $\widehat\varphi_\local^1$ was modified and $a_1$ must now regularly satisfy $\atbase_1$ for $3$ turns.
    Similar to the case where we loosened requirements, this can affect the obligations in a way opposite to the traditional case, because of the optimality we enforce.
    Here, strengthening the requirements makes achieving the previous value impossible for some environment inputs, which excuses the system from stronger obligations for those inputs.

    However, an existing implementation will fail since it violates $\widehat\varphi_\local^1$ where it previously did not.
    The new subformula of $\widehat\varphi_\local^1$ reads 
    \begin{align*}
        \LTLglobally \LTLfinally (\atbase_1 \land \LTLnext \atbase_1 \land \LTLnext \LTLnext \atbase_1).
    \end{align*}
    Because of the mismatch between the contract and $\widehat\varphi_\local^1$, the previous implementation no longer works.
    This is because $a_1$ is expected to both achieve satisfaction value $1$ for $\varphi_\local^1$ and respond to calls after two turns.
    Thus,  the contract is no longer suitable  and requires a revision.
\end{example}

\subsubsection{Procedure for Maximizing Local Satisfaction}

Consider a compositional  specification $\Phi$,   
a \gedc\ $(\alldec,\agv)$ for $\Phi$, 
and a MAS $\mc S = (\inp,  \outp,\langle M_1,\ldots,M_n\rangle)$ such that $\mc S \models (\alldec,\agv)$.
Suppose that the local specification for agent $a$ is modified from ${\varphi_\local^a}$ to  $\widehat\varphi_\local^a$, resulting in a revised global specification $\widehat\Phi$.
Let $\widehat M_a$ be a new implementation of agent $a$,  and 
$\widehat{\mc S} = (\inp,  \outp,\langle M_1,\ldots\widehat M_a,\ldots,M_n\rangle)$ the resulting MAS.
In order to check whether $\widehat{\mc S} \models (\alldec,\agv)$,  it suffices to verify only the conditions in \Cref{def:contract-sat} as pertaining to agent $a$,  
with respect to the new implementation $\widehat M_a$ of agent $a$ and the original local specification $\varphi_\local^a$.  
If this is the case,  by \Cref{thm:decomp-soundness},  we can conclude that the modified system $\widehat{\mc S}$ still satisfies the original specification $\Phi$. 
Otherwise,  the new implementation $\widehat M_a$ of agent $a$ is incompatible with the given contract.\looseness=-1

To further strengthen obligations on agent $a$ beyond the existing contract,
we additionally check if the new implementation maximizes the value of $\widehat\varphi_\local^a$ while satisfying the existing contract. 
That is,  whether $\widehat M_a$ satisfies conditions 
\ref{eq:local-ge-1},  \ref{eq:local-ge-2}  and \ref{eq:local-ge-3}, and in addition the property below holds:

For all $u \in \Vals{\widehat\varphi_\local^a}$,  if $\sigma_{\inp} \in \Hopeful{u,\widehat\varphi_\local^a}$,  then it holds that
 $\sema{\widehat\varphi_\local^a,(\sigma_{\inp} \parallel M_a(\sigma_{\inp} \parallel \sigma_{{\outnota}})\parallel \sigma_{{\outnota}})} \geq u$.

The precise property of a new implementation $\widehat M_a$ must then be as follows:
For every input sequence $ \sigma_{\inp} $, output of other components ${\sigma_{{\outnota}}}$, the following conditions are met:
\begin{itemize}
\item If there exists an output $\sigma_{\outa}$ where 
\[\sema{\Phi, ( \sigma_{\inp} \parallel \sigma_{{\outnota}} \parallel \sigma_{\outa})} = v,\]
then $\sema{\Phi, (\sigma_{\inp} \parallel \sigma_{\outnota} \parallel \widehat M_a (\sigma_{\inp} \parallel \sigma_{\outnota}))} \geq v$.
\item It here exists an output $\sigma_{\outa}$ where 
\[\sema{ {\widehat\varphi_\local^a} \cap G_a^v, ( \sigma_{\inp} \parallel \sigma_{\outnota} \parallel \sigma_{\outa})} = w,\]
then $\sema{{\widehat\varphi_\local^a}, (\sigma_{\inp} \parallel \sigma_{\outnota} \parallel \widehat M_a (\sigma_{\inp} \parallel \sigma_{\outnota}))} \geq w$.
\end{itemize}

This ensures that satisfaction of the modified local specification $\widehat\varphi_\local^a$ is maximized while the given contract  is satisfied.
The extended condition will fail if the new implementation does not satisfy the modified specification or if it violates the contract.

\newpage
\section{Benchmarks}
We implemented our method in a prototype,  which takes a compositional $\ltlf$ specification, an implementation model $\mc S$ and can optionally take a contract $(\alldec, \agv)$ as input.
The tool will check if the given implementation for each agent satisfies the conditions of contract satisfaction as defined in this paper.
It uses \textit{{S}pot}~\cite{duret.22.cav} (v2.10.6) for the automata operations.
To our knowledge, there are currently no other tools available that would apply to our setting or could easily be extended to compositional \ltlf\ model checking with good-enough satisfaction.

We compare the compositional approach with a monolithic method using our prototype, demonstrating the benefits of compositional methods for the same underlying implementation.
We performed experiments on several examples on a laptop with an Intel Core i7 processor at 2.8 GHz and 16 GB of memory.
We present the results of the evaluation in the main paper.
Below are the examples used in the experimental evaluation, including the contracts used in each case.

\subsubsection{delivery\_vehicles}
The intro example from the paper, a 3-agent-system that has the environment send \call\ and \badw\ via its input atomic proposition.
If there is a \call , agents need to respond by moving to $\atsite_a$ for the satisfaction of $\Psi_\shared$.
If there is \badw\ when $a_1$ or $a_2$ are traveling, their satisfaction value locally will be reduced.

\begin{itemize}
	\item $\Phi = \frac{1}{6} \varphi_1 \oplus \frac{1}{6} \varphi_2 \oplus \frac{1}{6} \varphi_3 \oplus \frac{1}{2} \psi_\shared$
	\item $I = \{ \call, \badw \}$
	\item $O_1 = \{ \atbase_1, \atsite_1 \}$ 
	\item $O_2 = \{ \atbase_2, \atsite_2 \}$
	\item $O_3 = \{ \atsite_3 \}$
\end{itemize}

For $i = 1,2$:
\begin{align*}
	\varphi_\local^i = &\LTLglobally (\lnot \atbase_i \lor \lnot \atsite_i) \land \\
	&\LTLglobally (\atbase_i \to \LTLnext \lnot \atsite_i) \land \\
	&\LTLglobally \LTLfinally (\atbase_i \land \LTLnext \atbase_i) \land \\
	&\LTLglobally (\badw \to \LTLnext ( (\atbase_i \lor \atsite_i) \oplus_{\frac{1}{2}} \top )). 
\end{align*}

\begin{align*}
\varphi_\local^3 &= \LTLglobally \LTLfinally \lnot \atsite_3
\end{align*}

\begin{align*}
	\Psi_\shared = \LTLglobally \call \to ((\lnot \atsite_1 \land \LTLnext (\atsite_1 \land \atsite_3)) \\
	\lor (\lnot \atsite_2 \land \LTLnext (\atsite_2 \land \atsite_3)) \\
	\lor \LTLnext (\lnot \atsite_1 \land \LTLnext (\atsite_1 \land \atsite_3)) \\
	\lor \LTLnext (\lnot \atsite_2 \land \LTLnext (\atsite_2 \land \atsite_3))).
\end{align*}

\textbf{Decomposition Contract}\\
$\alldec: \{(1,H_1,\{(1,1,1,1)\}),\\
(\frac{11}{12},H_\frac{11}{12},\{(0.5,1,1,1),(1,0.5,1,1)\}),\\
(\frac{5}{6},H_\frac{5}{6},\{(0,1,1,1),(1,0,1,1)\})\}$

Contract for agent 1:
\begin{itemize}
	\item $A_1^1 = \LTLglobally ( ( \atsite_1 \land (\LTLnext \atbase_1 ) \\
	\to \LTLnext( \call \to \LTLnext \LTLnext(\atsite_2))) \\ 
	\land ((\call \land \LTLnext \call) \to \LTLnext \LTLnext \atsite_3) \\
	\land ((\lnot \call \land \LTLnext \call) \to \LTLnext \LTLnext \LTLnext \atsite_3)) $

	\item $G_1^1 = \LTLglobally ( \atsite_2 \land (\LTLnext \atbase_2 ) \\
	\to \LTLnext( \call \to \LTLnext \LTLnext(\atsite_1)))$

	\item $A_1^{\frac{11}{12}} =  \LTLglobally ( ( ( \atsite_1 \land (\LTLnext \atbase_1 ) \\
	\to \LTLnext( \call \to \LTLnext \LTLnext(\atsite_2))) \\ 
	\land ((\call \land \LTLnext \call) \to \LTLnext \LTLnext \atsite_3) \\
	\land ((\lnot \call \land \LTLnext \call) \to \LTLnext \LTLnext \LTLnext \atsite_3) \\
	\land ((\lnot \atbase_2 \land \lnot \atsite_2 \land \badw) \\
	\to	\LTLglobally ( \call \land \badw \to \LTLnext \LTLnext(\atsite_2))))) $

	\item $G_1^{\frac{11}{12}} =  \LTLglobally ( \atsite_2 \land (\LTLnext \atbase_2 ) \\
	\to \LTLnext( \call \to \LTLnext \LTLnext(\atsite_1))\\
	\land ((\lnot \atbase_1 \land \lnot \atsite_1 \land \badw) \\
	\to	\LTLglobally ( \call \land \badw \to \LTLnext \LTLnext(\atsite_1))))$

	\item $A_1^{\frac{5}{6}} = \LTLglobally ( ( ( \atsite_1 \land (\LTLnext \atbase_1 ) \\
	\to \LTLnext( \call \to \LTLnext \LTLnext(\atsite_2))) \\ 
	\land ((\call \land \LTLnext \call) \to \LTLnext \LTLnext \atsite_3) \\
	\land ((\lnot \call \land \LTLnext \call) \to \LTLnext \LTLnext \LTLnext \atsite_3) \\
	\land ((\lnot \atbase_2 \land \lnot \atsite_2 \land \badw) \\
	\to	\LTLglobally ( \call \land \badw \to \LTLnext \LTLnext(\atsite_2)))))$

	\item $G_1^{\frac{5}{6}} = \LTLglobally ( \atsite_2 \land (\LTLnext \atbase_2 ) \\
	\to \LTLnext( \call \to \LTLnext \LTLnext(\atsite_1))\\
	\land ((\lnot \atbase_1 \land \lnot \atsite_1 \land \badw) \\
	\to	\LTLglobally ( \call \land \badw \to \LTLnext \LTLnext(\atsite_1))))$
\end{itemize}

Contract for agent 2:
\begin{itemize}
	\item $A_2^1 = \LTLglobally ( ( \atsite_2 \land (\LTLnext \atbase_2 ) \\
	\to \LTLnext( \call \to \LTLnext \LTLnext(\atsite_1))) \\ 
	\land ((\call \land \LTLnext \call) \to \LTLnext \LTLnext \atsite_3) \\
	\land ((\lnot \call \land \LTLnext \call) \to \LTLnext \LTLnext \LTLnext \atsite_3)) $

	\item $G_2^1 = \LTLglobally ( \atsite_1 \land (\LTLnext \atbase_1 ) \\
	\to \LTLnext( \call \to \LTLnext \LTLnext(\atsite_2)))$

	\item $A_2^{\frac{11}{12}} =  \LTLglobally ( ( ( \atsite_2 \land (\LTLnext \atbase_2 ) \\
	\to \LTLnext( \call \to \LTLnext \LTLnext(\atsite_1))) \\ 
	\land ((\call \land \LTLnext \call) \to \LTLnext \LTLnext \atsite_3) \\
	\land ((\lnot \call \land \LTLnext \call) \to \LTLnext \LTLnext \LTLnext \atsite_3) \\
	\land ((\lnot \atbase_1 \land \lnot \atsite_1 \land \badw) \\
	\to	\LTLglobally ( \call \land \badw \to \LTLnext \LTLnext(\atsite_1))))) $

	\item $G_2^{\frac{11}{12}} =  \LTLglobally ( \atsite_1 \land (\LTLnext \atbase_1 ) \\
	\to \LTLnext( \call \to \LTLnext \LTLnext(\atsite_2))\\
	\land ((\lnot \atbase_2 \land \lnot \atsite_2 \land \badw) \\
	\to	\LTLglobally ( \call \land \badw \to \LTLnext \LTLnext(\atsite_2))))$

	\item $A_2^{\frac{5}{6}} = \LTLglobally ( ( ( \atsite_2 \land (\LTLnext \atbase_2 ) \\
	\to \LTLnext( \call \to \LTLnext \LTLnext(\atsite_1))) \\ 
	\land ((\call \land \LTLnext \call) \to \LTLnext \LTLnext \atsite_3) \\
	\land ((\lnot \call \land \LTLnext \call) \to \LTLnext \LTLnext \LTLnext \atsite_3) \\
	\land ((\lnot \atbase_1 \land \lnot \atsite_1 \land \badw) \\
	\to	\LTLglobally ( \call \land \badw \to \LTLnext \LTLnext(\atsite_1)))))$
	
	\item $G_2^{\frac{5}{6}} = \LTLglobally ( \atsite_1 \land (\LTLnext \atbase_1 ) \\
	\to \LTLnext( \call \to \LTLnext \LTLnext(\atsite_2))\\
	\land ((\lnot \atbase_2 \land \lnot \atsite_2 \land \badw) \\
	\to	\LTLglobally ( \call \land \badw \to \LTLnext \LTLnext(\atsite_2))))$
\end{itemize}

Contract for agent 3:
\begin{itemize}
	\item $A_3^1 = \LTLglobally((\call \to \LTLnext \LTLnext(\atsite_1 \lor \atsite_2))) $
	
	\item $G_3^1 = \LTLglobally(((\call \land \LTLnext \call) \to \LTLnext \LTLnext \atsite_3) \\
	\land ((\lnot \call \land \LTLnext \call) \to \LTLnext \LTLnext \LTLnext \atsite_3))$

	\item $A_3^{\frac{11}{12}} =  \LTLglobally((\call \to \LTLnext \LTLnext(\atsite_1 \lor \atsite_2)))$, 
	
	\item $G_3^{\frac{11}{12}} =  \LTLglobally(((\call \land \LTLnext \call) \to \LTLnext \LTLnext \atsite_3) \\
	\land ((\lnot \call \land \LTLnext \call) \to \LTLnext \LTLnext \LTLnext \atsite_3))$

	\item $A_3^{\frac{5}{6}} = \LTLglobally((\call \to \LTLnext \LTLnext(\atsite_1 \lor \atsite_2)))$
	
	\item $G_3^{\frac{5}{6}} = \LTLglobally(((\call \land \LTLnext \call) \to \LTLnext \LTLnext \atsite_3) \\
	\land ((\lnot \call \land \LTLnext \call) \to \LTLnext \LTLnext \LTLnext \atsite_3))$
\end{itemize}
Modified specification:
\begin{align*}
	{\widehat{\varphi}_\local^1} &= \LTLglobally (\lnot \atbase_i \lor \lnot \atsite_i) \\
	&\land \LTLglobally (\atbase_i \to \LTLnext \lnot \atsite_i) \\
	&\land \LTLglobally \LTLfinally (\atbase_i \land \LTLnext \atbase_i). 
\end{align*}

\subsubsection{tasks\_collab}
\newcommand{\task}{\mathsf{task}}
\newcommand{\joint}{\mathsf{joint}}
\newcommand{\work}{\mathsf{work}}

A system with 4 agents. 
The environment uses the input atomic proposition $\task$ to force at least two agents into $\joint$ working for the next step.
The best local value for each agent $c$ is $\frac{1}{2}$, which is achieved by either doing $\work$ or being idle ($\lnot \work_c \land \lnot \joint_c$).
A decomposition contract is needed to determine which agents sacrifice their local satisfaction to deal with incoming $\task$.
In this case, agents 1 and 2 guarantee to handle $\Psi_\shared$, s.t. agents 3 and 4 can each satisfy $\varphi_\local^3$ and $\varphi_\local^4$ with the maximum value of $\frac{1}{2}$.
Importantly, if the environment is such that no $\task$ is issued, then all agents can focus on their local requirements.

\begin{itemize}
	\item $\Phi = \frac{1}{6} \varphi_\local^1 \oplus \frac{1}{6} \varphi_\local^2 \oplus \frac{1}{6} \varphi_\local^3 \oplus \frac{1}{6} \varphi_\local^4 \frac{1}{3} \oplus \Psi_\shared$
	\item $I = \{ \task \}$
	\item $O_c = \{ \work_c, \joint_c \}$
	\item $c \in \{ 1,2,3,4 \}$
\end{itemize}
\begin{align*}
\varphi_\local^c &= \LTLglobally ( (\LTLnext \work_c ) \oplus_{\frac{1}{2}} (\LTLnext (\lnot \work_c \land \lnot \joint_c)) )\\
 &\land \LTLglobally (\lnot \work_c \lor \lnot \joint_c) \\
\Psi_\shared &= \LTLglobally ( \task \to \LTLnext (\bigvee_{i \neq j} \joint_i \land \joint_j))
\end{align*}
\textbf{Decomposition Contract}\\
$\alldec$: $\{(\frac{2}{3},H_\frac{2}{3},\{(\frac{1}{2},\frac{1}{2},\frac{1}{2},\frac{1}{2},1)\}),(\frac{1}{2}, H_\frac{1}{2}, \{(0,0,\frac{1}{2},\frac{1}{2},1)\})\}$

Contract for agent 1:
\begin{itemize}
	\item $A_c^{\frac{2}{3}} =  \texttt{true}$, $G_c^{\frac{2}{3}} =  \texttt{true}$
	\item $A_1^{\frac{1}{2}} = \LTLglobally (\task \to \LTLnext(\joint_2))$
	\item $G_1^{\frac{1}{2}} = \LTLglobally(\task \to \LTLnext(\joint_1))$
\end{itemize}

Contract for agent 2:
\begin{itemize}
	\item $A_c^{\frac{2}{3}} =  \texttt{true}$, $G_c^{\frac{2}{3}} =  \texttt{true}$
	\item $A_2^{\frac{1}{2}} = \LTLglobally (\task \to \LTLnext(\joint_1))$
	\item $G_2^{\frac{1}{2}} = \LTLglobally(\task \to \LTLnext(\joint_2))$
\end{itemize}

Contract for agent 3:
\begin{itemize}
	\item $A_c^{\frac{2}{3}} =  \texttt{true}$, $G_c^{\frac{2}{3}} =  \texttt{true}$
	\item $A_3^{\frac{1}{2}} = \LTLglobally (\task \to \LTLnext(\joint_1 \land \joint_2))$
	\item $G_3^{\frac{1}{2}} = \texttt{true}$
\end{itemize}

Contract for agent 4:
\begin{itemize}
	\item $A_c^{\frac{2}{3}} =  \texttt{true}$, $G_c^{\frac{2}{3}} =  \texttt{true}$
	\item $A_4^{\frac{1}{2}} = \LTLglobally (\task \to \LTLnext(\joint_1 \land \joint_2))$
	\item $G_4^{\frac{1}{2}} = \texttt{true}$
\end{itemize}
Modified specification:
\begin{align*}
	{\widehat{\varphi}_\local^1} &= \LTLglobally ( (\LTLnext \work_1 ) \oplus_{\frac{1}{2}} (\LTLnext (\lnot \work_1)) )\\
	&\land \LTLglobally (\lnot \work_1 \lor \lnot \joint_1) \\
\end{align*}

\subsubsection{tasks\_scaled}

Direct extension to previous example with 6 agents. 
The environment uses the input atomic proposition $\task$ to force at least two agents into $\joint$ working for the next step.
The best local value for each agent is still $\frac{1}{2}$, which is achieved by either doing $\work$ or being idle ($\lnot \work_c \land \lnot \joint_c$).
A decomposition contract is needed to determine which agents devalue their local satisfaction to deal with incoming $\task$.
Here, agents 1 and 3 guarantee to handle $\Psi_\shared$, s.t. agents 2,4,5,6 can each satisfy $\varphi_\local^a$ with the maximum value of $\frac{1}{2}$.

\begin{itemize}
	\item $\Phi = \frac{1}{8} \varphi_\local^1 \oplus \frac{1}{8} \varphi_\local^2 
	\oplus \frac{1}{8} \varphi_\local^3 \oplus \frac{1}{8} \varphi_\local^4 
	\oplus \frac{1}{8} \varphi_\local^5 \oplus \frac{1}{8} \varphi_\local^6 \frac{1}{4} \oplus \Psi_\shared$
	\item $I = \{ \task \}$
	\item $O_c = \{ \work_c, \joint_c \}$
	\item $c \in \{ 1,2,3,4 \}$
\end{itemize}
\begin{align*}
\varphi_\local^c &= \LTLglobally ( (\LTLnext \work_c ) \oplus_{\frac{1}{2}} (\LTLnext (\lnot \work_c \land \lnot \joint_c)) )\\
 &\land \LTLglobally (\lnot \work_c \lor \lnot \joint_c) \\
\Psi_\shared &= \LTLglobally ( \task \to \LTLnext (\\
&\bigvee_{(i,j) \in \{ (1,3),(1,5),(3,5),(2,4),(2,6),(4,6)\}} \joint_i \land \joint_j))
\end{align*}
\textbf{Decomposition Contract}\\
$\alldec : \{(\frac{5}{8},H_\frac{5}{8},\{(\frac{1}{2},\frac{1}{2},\frac{1}{2},\frac{1}{2},\frac{1}{2},\frac{1}{2},1)\}),\\
(\frac{1}{2}, H_\frac{1}{2}, \{(0,\frac{1}{2},0,\frac{1}{2},\frac{1}{2},\frac{1}{2},1)\})\}$

Contract for agent 1:
\begin{itemize}
	\item $A_c^{\frac{5}{8}} =  \texttt{true}$, $G_c^{\frac{5}{8}} =  \texttt{true}$
	\item $A_1^{\frac{1}{2}} = \LTLglobally (\task \to \LTLnext(\joint_3))$
	\item $G_1^{\frac{1}{2}} = \LTLglobally(\task \to \LTLnext(\joint_1))$
\end{itemize}

Contract for agent 2:
\begin{itemize}
	\item $A_c^{\frac{5}{8}} =  \texttt{true}$, $G_c^{\frac{5}{8}} =  \texttt{true}$
	\item $A_3^{\frac{1}{2}} = \LTLglobally (\task \to \LTLnext(\joint_1 \land \joint_3))$
	\item $G_3^{\frac{1}{2}} = \texttt{true}$
\end{itemize}

Contract for agent 3:
\begin{itemize}
	\item $A_c^{\frac{5}{8}} =  \texttt{true}$, $G_c^{\frac{5}{8}} =  \texttt{true}$
	\item $A_2^{\frac{1}{2}} = \LTLglobally (\task \to \LTLnext(\joint_1))$
	\item $G_2^{\frac{1}{2}} = \LTLglobally(\task \to \LTLnext(\joint_3))$
\end{itemize}

Contract for agent 4:
\begin{itemize}
	\item $A_c^{\frac{5}{8}} =  \texttt{true}$, $G_c^{\frac{5}{8}} =  \texttt{true}$
	\item $A_4^{\frac{1}{2}} = \LTLglobally (\task \to \LTLnext(\joint_1 \land \joint_3))$
	\item $G_4^{\frac{1}{2}} = \texttt{true}$
\end{itemize}

Contract for agent 5:
\begin{itemize}
	\item $A_c^{\frac{5}{8}} =  \texttt{true}$, $G_c^{\frac{5}{8}} =  \texttt{true}$
	\item $A_5^{\frac{1}{2}} = \LTLglobally (\task \to \LTLnext(\joint_1 \land \joint_3))$
	\item $G_5^{\frac{1}{2}} = \texttt{true}$
\end{itemize}

Contract for agent 6:
\begin{itemize}
	\item $A_c^{\frac{5}{8}} =  \texttt{true}$, $G_c^{\frac{5}{8}} =  \texttt{true}$
	\item $A_6^{\frac{1}{2}} = \LTLglobally (\task \to \LTLnext(\joint_1 \land \joint_3))$
	\item $G_6^{\frac{1}{2}} = \texttt{true}$
\end{itemize}
Modified specification:
\begin{align*}
	{\widehat{\varphi}_\local^1} &= \LTLglobally ( (\LTLnext \work_1 ) \oplus_{\frac{1}{2}} (\LTLnext (\lnot \work_1)) )\\
	&\land \LTLglobally (\lnot \work_1 \lor \lnot \joint_1) \\
\end{align*}

\subsubsection{robots\_help\_a3}
\newcommand{\req}{\mathsf{req}}
\newcommand{\complete}{\mathsf{complete}}
\newcommand{\help}{\mathsf{help}}
\newcommand{\proceed}{\mathsf{proceed}}
System with 3 agents.
Here, the first agent has a higher weight to its local specification $\varphi_\local^1$. 
Thus the optimal strategy sees agents 2 and 3 sacrificing their satisfaction value to help agent 1 if the environment does not allow for full satisfaction.

\begin{itemize}
	\item $\Phi = (\varphi_\local^1 \oplus_{\frac{3}{5}} ( \varphi_\local^2 \oplus_{\frac{1}{2}} \varphi_\local^3 ) ) \oplus_{\frac{1}{2}} \Psi_\shared$
	\item $I = \{ \req \}$
	\item $O_1 = \{ \complete \}$
	\item $O_2 = \{ \help_2, \proceed_2 \}$
	\item $O_3 = \{ \help_3, \proceed_3 \}$
\end{itemize}
\begin{align*}
\varphi_\local^1 &= \LTLglobally ( \req \to \LTLnext( \lnot \complete ) ) \\
\varphi_\local^2 &= (\LTLglobally \LTLfinally \proceed_2) \land \LTLglobally (\lnot \help_2 \lor \lnot \proceed_2) \\
\varphi_\local^3 &= (\LTLglobally \LTLfinally \proceed_3) \land \LTLglobally (\lnot \help_3 \lor \lnot \proceed_3) \\
\Psi_\shared &= \LTLglobally (\complete \to (\help_2 \land \help_3)) \land \\
&\LTLglobally ((\lnot \complete) \to \LTLnext (\lnot \proceed_2 \land \lnot \proceed_3))
\end{align*}
\textbf{Decomposition Contract}\\
$\alldec$: $\{(1,H_1,\{(1,1,1,1)\}),(\frac{4}{5}, H_\frac{4}{5}, \{(1,0,0,1)\})\}$

Contract for agent 1:
\begin{itemize}
\item $A_1^1 = \LTLglobally (((\lnot \req \land \lnot \complete) \to \LTLnext(\help_2 \land \help_3)) \land (\lnot \complete \to \LTLnext (\lnot \proceed_2 \land \lnot \proceed_3))) $
\item $G_1^1 = (\lnot \complete) \land \LTLglobally((\lnot \req \land \lnot \complete) \to \LTLnext(\complete \land \LTLnext(\lnot \complete)))$
		
\item $A_1^{\frac{4}{5}} = \LTLglobally(((\lnot \req \land \lnot \complete) \to \LTLnext(\help_2 \land \help_3)) \land (\lnot \complete \to \LTLnext (\lnot \proceed_2 \land \lnot \proceed_3)))$
\item $G_1^{\frac{4}{5}} = (\lnot \complete) \land \LTLglobally((\lnot \req \land \lnot \complete) \to \LTLnext(\complete \land \LTLnext(\lnot \complete)))$
\end{itemize}

Contract for agent 2:
\begin{itemize}
\item $A_2^1 = \LTLglobally(((\lnot \req \land \lnot \complete) \to \LTLnext(\complete \land \help_3)) \land (\lnot \complete \to \LTLnext (\lnot \proceed_3)) \land ((\complete \lor \req) \to \LTLnext (\lnot \complete)))$
\item $G_2^1 = \LTLglobally(((\lnot \req \land \lnot \complete) \to \LTLnext(\help_2 )) \land (\lnot \complete \to \LTLnext (\lnot \proceed_2)))$

\item $A_2^{\frac{4}{5}} = \LTLglobally(((\lnot \req \land \lnot \complete) \to \LTLnext(\complete \land \help_3)) \land (\lnot  \complete \to \LTLnext (\lnot \proceed_3)) \land ((\complete \lor \req) \to \LTLnext (\lnot \complete)))$
\item $G_2^{\frac{4}{5}} = \LTLglobally(((\lnot \req \land \lnot \complete) \to \LTLnext(\help_2 )) \land (\lnot \complete \to \LTLnext (\lnot \proceed_2)))$
\end{itemize}

Contract for agent 3:
\begin{itemize}
\item $A_3^1 = \LTLglobally(((\lnot \req \land \lnot \complete) \to \LTLnext(\complete \land \help_2)) \land (\lnot \complete \to \LTLnext (\lnot \proceed_2)) \land ((\complete \lor \req) \to \LTLnext (\lnot \complete)))$
\item $G_3^1 = \LTLglobally(((\lnot \req \land \lnot \complete) \to \LTLnext(\help_3 )) \land (\lnot \complete \to \LTLnext (\lnot \proceed_3)))$

\item $A_3^{\frac{4}{5}} = \LTLglobally(((\lnot \req \land \lnot \complete) \to \LTLnext(\complete \land \help_2)) \land (\lnot \complete \to \LTLnext (\lnot \proceed_2)) \land ((\complete \lor \req) \to \LTLnext (\lnot \complete)))$
\item$G_3^{\frac{4}{5}} = \LTLglobally(((\lnot \req \land \lnot \complete) \to \LTLnext( \help_3)) \land (\lnot \complete \to \LTLnext (\lnot \proceed_3)))$
\end{itemize}
Modified specification:
\[
	{\widehat{\varphi}_\local^2} = (\LTLglobally \LTLfinally \proceed_2) \\
\]

\subsubsection{robots\_help\_a5}
System with 5 agents, scaled up from previous example.
Here, the first agent again has a higher weight to its local specification $\varphi_\local^1$. 
Thus the optimal strategy sees agents 2 and 3 sacrificing their satisfaction value to help agent 1 if the environment does not allow for full satisfaction.
Agents 4 and 5 are excused by the contract to focus on their own goals.

\begin{itemize}
	\item $\Phi = ({\frac{5}{15}} \varphi_\local^1 \oplus {\frac{1}{15}} \varphi_\local^2 \oplus {\frac{1}{15}} \varphi_\local^3  \oplus {\frac{1}{15}} \varphi_\local^4 \oplus {\frac{1}{15}} \varphi_\local^5 \oplus {\frac{6}{15}} \Psi_\shared$
	\item $I = \{ \req \}$
	\item $O_1 = \{ \complete \}$
	\item $O_2 = \{ \help_2, \proceed_2 \}$
	\item $O_3 = \{ \help_3, \proceed_3 \}$
	\item $O_4 = \{ \help_4, \proceed_4 \}$
	\item $O_5 = \{ \help_5, \proceed_5 \}$
\end{itemize}
\begin{align*}
\varphi_\local^1 &= \LTLglobally ( \req \to \LTLnext( \lnot \complete ) ) \\
\varphi_\local^2 &= (\LTLglobally \LTLfinally \proceed_2) \land \LTLglobally (\lnot \help_2 \lor \lnot \proceed_2) \\
\varphi_\local^3 &= (\LTLglobally \LTLfinally \proceed_3) \land \LTLglobally (\lnot \help_3 \lor \lnot \proceed_3) \\
\varphi_\local^4 &= (\LTLglobally \LTLfinally \proceed_4) \land \LTLglobally (\lnot \help_4 \lor \lnot \proceed_4) \\
\varphi_\local^5 &= (\LTLglobally \LTLfinally \proceed_5) \land \LTLglobally (\lnot \help_5 \lor \lnot \proceed_5) \\
\Psi_\shared &= \LTLglobally (\complete \to (\\
&\; \help_2 \land \help_3) \lor (\help_4 \land \help_5)) \land \\
&\LTLglobally ((\lnot \complete) \to \LTLnext (\lnot \proceed_2 \land \lnot \proceed_3 \\
&\; \land \lnot \proceed_4 \land \lnot \proceed_5))
\end{align*}
\textbf{Decomposition Contract}\\
$\alldec$: $\{(1,H_1,\{(1,1,1,1,1,1)\}),\\
(\frac{3}{4}, H_\frac{4}{5}, \{(1,0,0,1,1,1)\})\}$

Contract for agent 1:
\begin{itemize}
\item $A_1^1 = \LTLglobally (((\lnot \req \land \lnot \complete) \to \LTLnext(\help_2 \land \help_3)) \land (\lnot \complete \to \LTLnext (\lnot \proceed_2 \land \lnot \proceed_3))) $
\item $G_1^1 = (\lnot \complete) \land \LTLglobally((\lnot \req \land \lnot \complete) \to \LTLnext(\complete \land \LTLnext(\lnot \complete)))$
		
\item $A_1^{\frac{4}{5}} = \LTLglobally(((\lnot \req \land \lnot \complete) \to \LTLnext(\help_2 \land \help_3)) \land (\lnot \complete \to \LTLnext (\lnot \proceed_2 \land \lnot \proceed_3)))$
\item $G_1^{\frac{4}{5}} = (\lnot \complete) \land \LTLglobally((\lnot \req \land \lnot \complete) \to \LTLnext(\complete \land \LTLnext(\lnot \complete)))$
\end{itemize}

Contract for agent 2:
\begin{itemize}
\item $A_2^1 = \LTLglobally(((\lnot \req \land \lnot \complete) \to \LTLnext(\complete \land \help_3)) \land (\lnot \complete \to \LTLnext (\lnot \proceed_3)) \land ((\complete \lor \req) \to \LTLnext (\lnot \complete)))$
\item $G_2^1 = \LTLglobally(((\lnot \req \land \lnot \complete) \to \LTLnext(\help_2 )) \land (\lnot \complete \to \LTLnext (\lnot \proceed_2)))$

\item $A_2^{\frac{4}{5}} = \LTLglobally(((\lnot \req \land \lnot \complete) \to \LTLnext(\complete \land \help_3)) \land (\lnot  \complete \to \LTLnext (\lnot \proceed_3)) \land ((\complete \lor \req) \to \LTLnext (\lnot \complete)))$
\item $G_2^{\frac{4}{5}} = \LTLglobally(((\lnot \req \land \lnot \complete) \to \LTLnext(\help_2 )) \land (\lnot \complete \to \LTLnext (\lnot \proceed_2)))$
\end{itemize}

Contract for agent 3:
\begin{itemize}
\item $A_3^1 = \LTLglobally(((\lnot \req \land \lnot \complete) \to \LTLnext(\complete \land \help_2)) \land (\lnot \complete \to \LTLnext (\lnot \proceed_2)) \land ((\complete \lor \req) \to \LTLnext (\lnot \complete)))$
\item $G_3^1 = \LTLglobally(((\lnot \req \land \lnot \complete) \to \LTLnext(\help_3 )) \land (\lnot \complete \to \LTLnext (\lnot \proceed_3)))$

\item $A_3^{\frac{4}{5}} = \LTLglobally(((\lnot \req \land \lnot \complete) \to \LTLnext(\complete \land \help_2)) \land (\lnot \complete \to \LTLnext (\lnot \proceed_2)) \land ((\complete \lor \req) \to \LTLnext (\lnot \complete)))$
\item$G_3^{\frac{4}{5}} = \LTLglobally(((\lnot \req \land \lnot \complete) \to \LTLnext( \help_3)) \land (\lnot \complete \to \LTLnext (\lnot \proceed_3)))$
\end{itemize}

Contract for agent 4:
\begin{itemize}
\item $A_3^1 = \LTLglobally(((\lnot \req \land \lnot \complete) \to \LTLnext(\complete \land \help_2)) \land (\lnot \complete \to \LTLnext (\lnot \proceed_2)) \land ((\complete \lor \req) \to \LTLnext (\lnot \complete)))$
\item $G_3^1 = \LTLglobally(((\lnot \req \land \lnot \complete) \to \LTLnext(\help_3 )) \land (\lnot \complete \to \LTLnext (\lnot \proceed_3)))$

\item $A_3^{\frac{4}{5}} = \LTLglobally(((\lnot \req \land \lnot \complete) \to \LTLnext(\complete \land \help_2)) \land (\lnot \complete \to \LTLnext (\lnot \proceed_2)) \land ((\complete \lor \req) \to \LTLnext (\lnot \complete)))$
\item$G_3^{\frac{4}{5}} = \LTLglobally(((\lnot \req \land \lnot \complete) \to \LTLnext( \help_3)) \land (\lnot \complete \to \LTLnext (\lnot \proceed_3)))$
\end{itemize}

Contract for agent 5:
\begin{itemize}
\item $A_3^1 = \LTLglobally(((\lnot \req \land \lnot \complete) \to \LTLnext(\complete \land \help_2)) \land (\lnot \complete \to \LTLnext (\lnot \proceed_2)) \land ((\complete \lor \req) \to \LTLnext (\lnot \complete)))$
\item $G_3^1 = \LTLglobally(((\lnot \req \land \lnot \complete) \to \LTLnext(\help_3 )) \land (\lnot \complete \to \LTLnext (\lnot \proceed_3)))$

\item $A_3^{\frac{4}{5}} = \LTLglobally(((\lnot \req \land \lnot \complete) \to \LTLnext(\complete \land \help_2)) \land (\lnot \complete \to \LTLnext (\lnot \proceed_2)) \land ((\complete \lor \req) \to \LTLnext (\lnot \complete)))$
\item$G_3^{\frac{4}{5}} = \LTLglobally(((\lnot \req \land \lnot \complete) \to \LTLnext( \help_3)) \land (\lnot \complete \to \LTLnext (\lnot \proceed_3)))$
\end{itemize}

Modified specification:
\[
	{\widehat{\varphi}_\local^2} = (\LTLglobally \LTLfinally \proceed_2) \\
\]

\subsubsection{synchr\_response} 
\newcommand{\respond}{\mathsf{respond}}
A 3-agent-system that needs to respond to environment input atomic proposition $\call$ by $\respond$ at the same time.
Each agent needs to $\work$, where different lengths of consecutive work change the satisfaction of $\varphi_\local$.
If no $\call$ comes in, the best value overall is $1$, and each agent can either $\work$, or $\respond$.
Note that each agent is free to skip every other work, as the overall value is the minimum over all subformulas.

\begin{itemize}
	\item $\Phi = \varphi_1 \land \varphi_2 \land \varphi_3 \land \Psi_\shared$
	\item $I = \{ \call \}$
	\item $O_c = \{ \work_c, \respond_c \}$
	\item $c \in \{ 1, 2, 3 \}$
\end{itemize}

\begin{align*}
\varphi_\local^c &= \LTLglobally ( \work_c \oplus_{\frac{1}{3}} (\lnot \work_c \to \LTLnext \work_c ) \\
&\; \oplus_{\frac{1}{2}} (\texttt{true}))\\
&\; \land (\lnot \work_c \lor \lnot \respond_c)) \\
\Psi_\shared &= \LTLglobally (\call \land \LTLnext (\bigwedge_c \respond_c) \oplus_{\frac{1}{2}} \texttt{true} )
\end{align*}

\textbf{Decomposition Contract}\\
$\alldec$: $\{(\frac{1}{3},H_\frac{1}{3},\{(\frac{1}{3},\frac{1}{3},\frac{1}{3},1)\}),\\
(\frac{1}{2}, H_\frac{1}{2}, \{(\frac{2}{3},\frac{2}{3},\frac{2}{3},\frac{1}{2}),(1,\frac{2}{3},\frac{2}{3},\frac{1}{2}),\\
(\frac{2}{3},1,\frac{2}{3},\frac{1}{2}),\dots\})\}$

Contract for agent 1:
\begin{itemize}
	\item $A_1^{\frac{2}{3}} =  \texttt{true}$, $G_1^{\frac{2}{3}} =  \texttt{true}$
	\item $A_1^{\frac{1}{2}} = \LTLglobally (\call \to \LTLnext(\respond_2 \land \respond_3))$
	\item $G_1^{\frac{1}{2}} = \LTLglobally(\call \to \LTLnext(\respond_1))$
\end{itemize}

Contract for agent 2:
\begin{itemize}
	\item $A_2^{\frac{2}{3}} =  \texttt{true}$, $G_2^{\frac{2}{3}} =  \texttt{true}$
	\item $A_2^{\frac{1}{2}} = \LTLglobally (\call \to \LTLnext(\respond_1 \land \respond_3))$
	\item $G_2^{\frac{1}{2}} = \LTLglobally(\call \to \LTLnext(\respond_2))$
\end{itemize}

Contract for agent 3:
\begin{itemize}
	\item $A_3^{\frac{2}{3}} =  \texttt{true}$, $G_3^{\frac{2}{3}} =  \texttt{true}$
	\item $A_3^{\frac{1}{2}} = \LTLglobally (\call \to \LTLnext(\respond_1 \land \respond_2))$
	\item $G_3^{\frac{1}{2}} = \LTLglobally(\call \to \LTLnext(\respond_3))$
\end{itemize}

Modified specification:
\begin{align*}
	{\widehat{\varphi}_\local^1} &= \LTLglobally ( \work_1 \oplus_{\frac{1}{3}} (\LTLnext \work_1 ) \\
	 &\oplus_{\frac{1}{2}} (\texttt{true}))\\
	 &\land (\lnot \work_1 \lor \lnot \respond_1) \\
\end{align*}

\end{document}